\documentclass{llncs}
\usepackage[T1]{fontenc}
\usepackage{lmodern}
\usepackage[utf8]{inputenc}
\usepackage{amsmath,amssymb}
\usepackage{amsfonts}
\usepackage{dsfont}
\usepackage{mathtools}
\usepackage[dvipsnames]{xcolor}
\usepackage{tikz}
\usepackage{comment}
\usepackage{enumitem}
\usepackage[colorlinks = true,
            linkcolor = blue,
            urlcolor  = blue,
            citecolor = blue,
            anchorcolor = blue]{hyperref}

\definecolor{myred}{rgb}{1,0.6,0.6}
\definecolor{myblue}{rgb}{0,0,1}

\usepackage{tabularx}
\usepackage{booktabs}

\newcommand{\Oh}{\ensuremath{\mathcal{O}}}
\DeclareMathOperator{\dist}{dist}
\DeclareMathOperator{\diam}{diam}

\spnewtheorem{rules}{Rule}{\upshape\bfseries}{\upshape\itshape}
\spnewtheorem{rrule}[theorem]{Reduction Rule}{\upshape\bfseries}{\upshape\itshape}
\spnewtheorem{brule}[theorem]{Branching Rule}{\upshape\bfseries}{\upshape\itshape}
\spnewtheorem{invar}[theorem]{Invariant}{\upshape\bfseries}{\upshape\itshape}
\spnewtheorem{const}[theorem]{Construction}{\upshape\bfseries}{\upshape\rmfamily}
\spnewtheorem{observation}[theorem]{Observation}{\upshape\bfseries}{\upshape\itshape}

\spnewtheorem{claim}{Claim}{\upshape\itshape}{\upshape\rmfamily}
\newenvironment{claimproof}{{\noindent\textit{Proof. }}}{\hfill$\blacksquare$}

\newcommand{\trisclub}{\textsc{Vertex Triangle $s$-Club}}

\newcommand{\etrisclub}{\textsc{Edge Triangle $s$-Club}}
\newcommand{\seedsclub}{\textsc{Seeded $s$-Club}}

\newcommand{\problemdef}[3]{
  \begin{center}
    \begin{minipage}{0.95\textwidth}
      \normalsize\textsc{#1} \smallskip \\
      \begin{tabularx}{\textwidth}{@{}l@{\hspace{3pt}}X}
        \normalsize\textbf{Input:}    & \normalsize#2 \\
        \normalsize\textbf{Question:} & \normalsize#3
      \end{tabularx}
    \end{minipage}
  \end{center}
}

\newcolumntype{y}{>{\hsize=.8\hsize}X}
\newcolumntype{t}{>{\hsize=.6\hsize}X}

\begin{document}

\title{The Parameterized Complexity of~$s$-Club with Triangle and Seed Constraints }
\titlerunning{The Parameterized Complexity of~$s$-Club with Triangle and Seed Constraints}
\authorrunning{J.~Garvardt, C.~Komusiewicz, F.~Sommer}

\author{Jaroslav Garvardt \and Christian Komusiewicz\orcidID{0000-0003-0829-7032} \and Frank Sommer \thanks{Supported by the  Deutsche Forschungsgemeinschaft  (DFG), project EAGR (KO 3669/6-1).}
  \orcidID{0000-0003-4034-525X}}

\institute{Fachbereich Mathematik und Informatik, Philipps-Universität Marburg,\\ Marburg, Germany \\ \email{\{garvardt,komusiewicz,fsommer\}@informatik.uni-marburg.de}}










\maketitle

\begin{abstract}
The \textsc{s-Club} problem asks, for a given undirected graph~$G$, whether~$G$ contains a vertex set~$S$ of size at least $k$ such that~$G[S]$, the subgraph of~$G$ induced by~$S$, has diameter at most~$s$. 
We consider variants of \textsc{$s$-Club} where one additionally demands that each vertex of~$G[S]$ is contained in at least~$\ell$ triangles in~$G[S]$, that~$G[S]$ contains a spanning subgraph~$G'$ such that each edge of~$E(G')$ is contained in at least $\ell$~triangles in~$G'$, or that~$S$ contains a given set~$W$ of seed vertices. 
We show that in general these variants are W[1]-hard when parameterized by the solution size~$k$, making them significantly harder than the unconstrained~\textsc{$s$-Club} problem. 
On the positive side, we obtain some FPT algorithms for the case when~$\ell=1$ and for the case when~$G[W]$, the graph induced by the set of seed vertices, is a clique.  
\end{abstract}

\section{Introduction}
Finding cohesive subgroups in social or biological networks is a fundamental task
in network analysis. A classic formulation of cohesiveness is based on the
observation that cohesive groups have small diameter. This observation led to the~$s$-club model originally proposed by Mokken~\cite{Mok79}. An \emph{$s$-club} in a graph~$G=(V,E)$ is a set of vertices~$S$ such that~$G[S]$, the subgraph of~$G$ induced by~$S$, has diameter at most~$s$. The 1-clubs are thus precisely the cliques and the larger the value of~$s$, the more the clique-defining constraint of having diameter one is relaxed. In the \textsc{$s$-Club} problem we aim to decide whether~$G$ contains an~$s$-club of size at least~$k$.

A big drawback of $s$-clubs is that the largest $s$-clubs are often not very cohesive with
respect to other cohesiveness measures such as density or minimum degree. This behavior is particularly pronounced for~$s=2$: the largest
$2$-club in a graph is often the vertex~$v$ of maximum degree together with its neighbors~\cite{HKN15}.
To avoid  these so-called hub-and-spoke structures, it has been proposed to
augment the $s$-club definition with additional constraints~\cite{CA17,KNNP19,PYB13,VB12}.

One of these augmented models, proposed by Carvalho and Almeide~\cite{CA17}, asks that every vertex
is part of a triangle~\cite{CA17}. This property was later generalized to the \emph{vertex-$\ell$-triangle} property, which asks that every vertex of~$S$ is in at
least~$\ell$ triangles in~$G[S]$~\cite{AB19}. \problemdef{Vertex Triangle~$s$-Club} {An undirected graph~$G$ = $(V,E)$, and
  two integers~$k,\ell\geq 1$.}  {Does~$G$ contain an $s$-club~$S$ of size at least~$k$
  that fulfills the vertex-$\ell$-triangle property?}
  
 The vertex-$\ell$-triangle
constraint entails some desirable properties for cohesive subgraphs. For example, in a
vertex-$\ell$-triangle~$s$-club, the minimum degree is larger
than~$\sqrt{2\ell}$. However, some undesirable behavior of hub-and-spoke structures
remains. For example, the graph consisting of two cliques of size~$d+1$ that are connected
via one edge is a vertex-$\binom{d}{2}$-triangle 3-club but it can be made disconnected
via one edge deletion. Thus, vertex-$\ell$-triangle $s$-clubs are not robust with respect to edge deletions.

To overcome this problem, we introduce a new model where we  put triangle constraints on the edges of the~$s$-club instead of the vertices. 
More precisely, we say that a vertex set~$S$ of a graph~$G$ fulfills the \emph{edge-$\ell$-triangle} property if~$G[S]$ contains a spanning subgraph~$G'\coloneqq (S,E')$ such that every edge in~$E(G')$ is in at least~$\ell$ triangles in~$G'$ and the diameter of~$G'$ is at most~$s$.
Next, we introduce the related problem.

\problemdef{Edge Triangle~$s$-Club}
{An undirected graph~$G=(V,E)$, and two integers~$k,\ell\geq 1$.}
{Does~$G$ contain a vertex set~$S$ of size at least~$k$ that fulfills the edge-$\ell$-triangle property?} 

Note that in this definition, the triangle and diameter constraints are imposed on a spanning subgraph of~$G[S]$.
In contrast, for \textsc{Vertex Triangle~$s$-Club}, they are imposed directly on~$G[S]$.
The reason for this distinction is that we would like to have properties that are closed under edge insertions.
Properties which are closed under edge insertions are also well-motivated from an application point of view since adding a new connection within a group should not destroy this group.
If we would impose the triangle constraint on the induced subgraph~$G[S]$ instead, then an edge-$\ell$-triangle~$s$-club~$S$ would not be robust to edge additions.
For example, consider a graph~$G$ consisting of clique~$C$ to which two vertices~$u$ and~$v$ are attached in such a way that both~$u$ and~$v$ have exactly $2$~neighbors in~$C$ which are distinct.
The~$V(G)$ is an edge-$1$-triangle~$3$-club, but other adding the edge~$uv$, the edge~$uv$ is contained in no triangle and thus any edge-$1$-triangle~$3$-club cannot contain both~$u$ and~$v$.

Observe that every set that fulfills the  edge-$\ell$-triangle property also fulfills the vertex-$\ell$-triangle property. 
Also note that the converse is not true: 
A vertex-$\ell$-triangle~$s$-club is not necessarily also an edge-$\ell$-triangle~$s$-club.
For this, consider the above mentioned graph consisting of two cliques of size~$d+1$ that are connected
via one edge
Then, this graph is an~vertex-$\binom{d}{2}$-triangle-$3$-club. But after deleting~$e$ the graph is disconnected.
Moreover, each vertex~$v\in S$ has at least~$\ell+1$ neighbors in~$S$:
Consider an arbitrary edge~$uv$. 
Since~$uv$ is in at least~$\ell$ triangles~$\{u,v,w_1\},\ldots,\{u,v,w_\ell\}$ we thus conclude that~$u$ and~$v$ have degree at least~$\ell$.
We can show an even stronger statement: an edge-$\ell$-triangle $ s $-club $ S $ is robust against up to $ \ell $ edge deletions, as desired.

\begin{proposition}
\label{lem:edge-triangle-club-robust-edge-deletions}
Let~$G=(V,E)$ be a graph and let~$S$ be an edge-$\ell$-triangle $ s $-club in~$G$.
More precisely, let~$G'$ be a spanning subgraph of~$G[S]$ such that every edge in~$E(G')$ is in at least~$\ell$ triangles in~$G'$ and the diameter of~$G'$ is at most~$s$.
If $ \ell $ edges are removed from~$G'$, then~$S$ is still an $ (s+\ell)$-club and a~$(2s)$-club.
\end{proposition}

\begin{proof}
	We show that if $ \ell $ edges are removed from $ G' $, the diameter of the resulting graph $ \widetilde{G} $ increases by at most $ \ell $. 
	Let $ P=(v_1,\ldots,v_{s+1}) $ be a path of length~$ s $ in $ G' $. 
	Since $ G' $ is an edge-$\ell$-triangle $ s $-club, every edge $ v_i v_{i+1} $ of~$P $ is part of at least $ \ell $~triangles in~$G'$. 
	Thus, for two vertices $ v_i$ and~$v_{i+1}$ in~$P$ there is a path of length at most two from $ v_i $ to $ v_{i+1} $ in $ G' $, either directly through the edge $ v_i v_{i+1}$ or via a vertex $ u $ that forms one of the $ \ell $ triangles with $ v_i $ and $ v_{i+1} $ in~$ G' $. 
	Thus, $ \dist(v_i,v_{i+1}) $ increases by at most $ 1 $ after one edge deletion and only if~$v_i v_{i+1}$ is removed.
	Since at most $ \ell $ of the edges in $ P $ are removed, we have $ \dist(v_1,v_{s+1}) \leq \dist(v_1,v_2) + \ldots + \dist(v_s,v_{s+1}) \leq s + \ell $ in $ \widetilde{G} $. 
	By the same arguments, we also have $\dist(v_1,v_{s+1}) \leq 2s$.
	
	Thus, after deleting $ \ell $ edges in $ G' $,~$S$ is an $ (s+\ell) $-club and a~$(2s)$-club.  	
	\qed
\end{proof}

The following further variant of \textsc{$s$-Club} is also practically motivated but not necessarily by concerns about the robustness of the $s$-club. Here the difference to the standard problem is simply that we are given a set of seed vertices~$W$ and aim to find a large $s$-club that contains all seed vertices.
\problemdef{Seeded~$s$-Club}
{An undirected graph~$G$ = $(V,E)$, a subset~$W\subseteq V$, and an integer~$k\geq 1$.}
{Does~$G$ contain an $s$-club~$S$ of size at least~$k$ such that~$W\subseteq S$?}

This variant has applications in community detection, where we are often interested in finding communities containing some set of fixed vertices~\cite{Kan14,WGD13}. 

In this work, we study the parameterized complexity of the three above-mentioned problems with
respect to the standard parameter solution size~$k$. Our goal is to determine whether FPT results for~\textsc{$s$-Club}~\cite{CHLS13,SKMN12} transfer to these practically motivated problem variants.

\paragraph{Known Results.}
The \textsc{$s$-Club} problem is NP-hard for all~$s\ge 1$~\cite{BLP02}, even when the input graph has diameter~$s+1$~\cite{BBT05}. For~$s=1$, \textsc{$s$-Club} is equivalent to \textsc{Clique} and thus W[1]-hard with respect to~$k$. 
In contrast, for every~$s>1$,
\textsc{$s$-Club} is fixed-parameter tractable (FPT) with respect to the solution
size~$k$~\cite{CHLS13,SKMN12}. 
 This fixed-parameter tractability  can be shown via a Turing kernel with $\Oh(k^2)$~vertices for even~$s$ and~$\Oh(k^3)$~vertices for odd~$s$~\cite{SKMN12,CHLS13}. The complexity of~\textsc{$s$-Club} has been also studied with respect to different classes of input graphs~\cite{GHKR14}  and with respect to~structural parameters such as degeneracy of the input graph~\cite{HKNS15}. The \textsc{$s$-Club} problem can be solved 
efficiently in practice, in particular for~$s=2$~\cite{BLP02,CHLS13,HKN15}.  In particular, the \textsc{$2$-Club} problem has efficient branch-and-bound algorithms~\cite{CHLS13,HKN15} which can compute the optimal solutions on very large sparse graphs.

\trisclub{} is NP-hard for all~$s\ge 1$ and for all~$\ell\ge 1$~\cite{CA17,AB19}. We are not aware of any algorithmic studies of \etrisclub{} or~\seedsclub. 
NP-hardness of \etrisclub{} for~$\ell=1$ can be shown via the reduction for \trisclub{} for~$\ell=1$~\cite{AB19}.
Also, the NP-hardness of~\seedsclub{} for~$W\ne\emptyset$ follows directly from the fact that an algorithm for the case where~$|W|=1$ can be used as a black box to solve \textsc{$s$-Club}. Further robust models of $s$-clubs, which are not considered in this work, include
$t$-hereditary $s$-clubs~\cite{PYB13}, $t$-robust $s$-clubs~\cite{VB12}, and $t$-connected $s$-clubs~\cite{YPB17,KNNP19}. For an
overview on clique relaxation models and complexity issues for the corresponding subgraph
problems we refer to the relevant surveys~\cite{K16,PYB13}.



\paragraph{Our Results.} An overview of our results is given in Table~\ref{tab:results}.
For \trisclub{} and \etrisclub{}, we provide a complexity dichotomy for all interesting combinations of~$s$ and~$\ell$, that is, for all~$s\ge 2$ and~$\ell\ge 1$,
into cases that are FPT or W[1]-hard with
respect to~$k$, respectively.  
Our W[1]-hardness reduction for \etrisclub{} for~$\ell\ge 2$ also shows the NP-hardness of this case. 
The FPT-algorithms are obtained via adaptions of the Turing kernelization for \textsc{$s$-Club}. 
Interestingly, \trisclub~ with~$\ell=1$  is
FPT only for larger~$s$, whereas \etrisclub~with~$\ell=1$ is FPT for all~$s$. 
In our opinion, this means that the edge-$\ell$-triangle property is preferable not only from a modelling standpoint but also from an algorithmic standpoint as it allows to employ Turing kernelization as a part of the solving procedure, at least for~$\ell=1$. 
It is easy to see that standard problem kernels of polynomial size are unlikely to exist for \trisclub{} and \etrisclub{}: $s$-clubs are necessarily connected and thus taking the disjoint union of graphs gives a trivial or-composition and, therefore, a polynomial problem kernel implies coNP $\subseteq$ NP/poly~\cite{BDFH09}.  

All of our hardness results for \trisclub{} and \etrisclub{} are shown by a reduction from \textsc{Clique} which is W[1]-hard with respect to~$k$~\cite{Cyg+15,DF13}.
The idea is to replace each vertex  of the \textsc{Clique} instance by a vertex gadget.
These gadgets are constructed in such a way that if one of these vertices is part of a vertex/edge-$\ell$-triangle~$2$-club~$S$, then the entire vertex gadget is part of~$S$.
We can then use the distance constraint to make sure that full vertex gadgets are chosen only if the corresponding vertices are adjacent.  
While this idea is very natural, using the triangle constraint without creating many vertices that are too close to each other turned out to be technically challenging.

For \seedsclub, we provide a kernel with respect to~$k$ for clique seeds~$W$ and W[1]-hardness with respect to~$k$ for some other cases. 
For~$s=2$, our results provide a dichotomy into FPT and W[1]-hardness with respect to~$k$ in terms of the structure of the seed.

\begin{table}[t]
  \caption{Overview of our results of the parameterized complexity of the three problems with respect to the parameter solution size~$k$.}
  \begin{tabularx}{\textwidth}{ltty}
    \toprule
     & \textsc{Vertex Triangle~~~~~\linebreak $s$-Club}~ & \textsc{Edge Triangle~~~~~\linebreak $s$-Club}~ &  \seedsclub \\
    \midrule
   FPT &\ $\ell=1$ and~$s\ge 4$ & $\ell=1$ for each~$s$ & $W$ is a clique\\
   \midrule
    W[1]-h &\ $\ell=1$ and $s\le 3$ & $\ell\ge 2$ for each~$s$ & $s=2$  and~$G[W]$ contains at least two non-adjacent vertices   \\[6.5ex]
    &\ $\ell\ge 2$ for each~$s$ & & $s\ge 3$ and $G[W]$ contains at\linebreak least~$2$ connected components\\
    \bottomrule
  \end{tabularx}
  \label{tab:results}
\end{table}

The W[1]-hardness of \seedsclub{} is provided by two reductions from \textsc{Clique}.
One reduction is for the case~$s=2$ and any seed that contains at least two non-adjacent vertices~$u$ and~$z$.
The other reduction  is for the case~$s\ge 3$ and any seed that contains at least two connected components~$U$ and~$Z$.
In both cases we add two copies~$X_1$ and~$X_2$ of the graph of the \textsc{Clique} instance to the new instance of \seedsclub{} such that each vertex of~$X_1$ has distance at most~$s$ to~$u$ or~$U$ if its copy in~$X_2$ is also part of the solution.
We show a similar property for~$X_2$ and~$z$ or~$Z$.
This ensures that if~$p_1\in X_1$ is in the solution if and only if~$p_2\in X_2$ is in the solution.
Furthermore, the reductions have the property that vertex~$p_1$ in~$X_1$ has distance at most~$s$ to vertex~$q_2$ in~$X_2$ if and only if~$pq$ is an edge of the \textsc{Clique} instance.
This feature will then ensure that the same clique has to be chosen from both copies.

Our W[1]-hardness results, in particular those for \seedsclub, show that the FPT results for \textsc{$s$-Club} are quite brittle since the standard argument that we may assume~$k\ge\Delta$ fails and that adding even simple further constraints makes finding small-diameter subgraphs much harder.

\paragraph{Preliminaries.}
For integers~$p,q$, we denote~$[p,q]\coloneqq \{p,p+1, \ldots ,q\}$ and~$[q]\coloneqq [1,q]$.
For a graph~$G$, we let~$V(G)$  denote its vertex set and~$E(G)$ its edge set.
We let~$n$ and~$m$ denote the order of~$G$ and the number of edges in~$G$, respectively.
A \emph{path of length~$p$} is a sequence of pairwise distinct vertices~$v_1,\ldots, v_{p+1}$ such that~$v_iv_{i+1}\in E(G)$ for each~$i\in[p]$.
The \emph{distance}~$\dist_G(u,v)$ is the length of a shortest path between vertices~$u$ and~$v$.
Furthermore, we define~$\dist_G(u,W)\coloneqq \min_{w\in W}\dist(u,w)$.
We denote by~$\diam_G(G)\coloneqq \max_{u,v\in V(G)}\dist_G(u,v)$ the \emph{diameter} of~$G$.
Let~$S\subseteq V(G)$ be a vertex set.
We denote by~$N_i(S)\coloneqq \{u \in V \mid \dist(u,S) = i \}$ 
the \emph{open~$i$-neighborhood} of~$S$ and by~$N_i[S]\coloneqq \bigcup_{j\le i}N_i(S)\cup S$ the \emph{closed~$i$-neighborhood} of~$S$. For a vertex~$v\in V(G)$, we write~$N_i(v)\coloneqq N_i(\{v\})$ and~$N_i[v]\coloneqq N_i[\{v\}]$. 
A graph~$G'\coloneqq (V',E')$ with~$V'\subseteq V$, and~$E'\subseteq E(G[V'])$ is a \emph{subgraph} of~$G$.
By~$G[S]\coloneqq (S,\{uv\in E(G)\mid u,v\in S\})$ we denote the \emph{subgraph induced by}~$S$.
Furthermore, by~$G-S\coloneqq G[V\setminus S]$ we denote the induced subgraph obtained after the deletion of the vertices in~$S$.
A vertex set such that each pair of vertices is adjacent is called a \emph{clique} and a clique consisting of three vertices is a \textit{triangle}. 

For the definitions of parameterized complexity theory, we refer to the standard monographs~\cite{Cyg+15,DF13}.
All of our hardness results are shown by a reduction from \textsc{Clique}.  
\problemdef{Clique}
{An undirected graph~$G=(V,E)$ and an integer~$k$.}
{Does~$G$ contain a clique of size at least~$k$?}
\textsc{Clique} is W[1]-hard with respect to~$k$~\cite{Cyg+15,DF13}.

\section{Vertex Triangle~\texorpdfstring{$s$}{s}-Club}
In this section, we settle the parameterized complexity of \textsc{Vertex Triangle~$s$-Club} with respect to the solution size~$k$.
First, we show that this problem is fixed-parameter tractable when~$\ell=1$ and~$s\ge 4$.
Afterwards, we show W[1]-hardness for all remaining cases, that is, for~$\ell\ge 2$ and~$s\ge 2$, and also for~$\ell=1$ and~$s\in\{2,3\}$.

\subsection{FPT-Algorithms}

The overall idea is based on the idea of the Turing kernel for \textsc{$2$-Club}, that is, bounding the size of~$N_s[v]$ for each vertex~$v\in V(G)$.
The first step is to remove all vertices which are not in a triangle.

\begin{rrule}
\label{rr-1-vertex-triangle-s-club-remove-vertices-in-no-triangle}
Let~$(G,k)$ be an instance of \textsc{Vertex Triangle~$s$-Club}. 
Delete all vertices from~$G$ which are not part of any triangle.
\end{rrule}

Clearly, Reduction Rule~\ref{rr-1-vertex-triangle-s-club-remove-vertices-in-no-triangle} is correct and can be exhaustively applied in polynomial time.
The application of Reduction Rule~\ref{rr-1-vertex-triangle-s-club-remove-vertices-in-no-triangle} has the following effect: if some vertex~$v$ is close to many vertices, then~$(G,k)$ is a trivial yes-instance.

\begin{lemma}
\label{lem-1-vertex-triangle-s-club-neighborhood}
Let~$(G,k)$ be an instance of \textsc{Vertex Triangle~$s$-Club} with~$\ell=1$ and~$s\ge 4$ to which Reduction Rule~\ref{rr-1-vertex-triangle-s-club-remove-vertices-in-no-triangle} is applied.
Then,~$(G,k)$ is a yes-instance if~$|N_{\lfloor s/2\rfloor-1}[v]|\ge k$ for some vertex~$v\in V(G)$.
\end{lemma}

\begin{proof}
Let~$v\in V(G)$ be a vertex such that~$|N_{\lfloor s/2\rfloor-1}[v]|\ge k$.
We construct a vertex-$1$-triangle~$s$-club~$T$ of size at least~$|N_{\lfloor s/2\rfloor-1}[v]|\ge k$.
Initially, we set~$T\coloneqq N_{\lfloor s/2\rfloor-1}[v]$.
Now, for each vertex~$w\in N_{\lfloor s/2\rfloor-1}(v)$ we do the following:
Since Reduction Rule~\ref{rr-1-vertex-triangle-s-club-remove-vertices-in-no-triangle} is applied, we conclude that there exist two vertices~$x$ and~$y$ such that~$G[\{ w,x,y\}]$ is a triangle.
We add~$x$ and~$y$ to the set~$T$. 
We call the set of vertices added in this step the \emph{$T$-expansion}.

Next, we show that~$T$ is indeed a vertex-$1$-triangle~$s$-club for~$s\ge 4$. 
Observe that each vertex in~$T$ is either in~$N_{\lfloor s/2\rfloor-1}[v]$ or a neighbor of a vertex in~$N_{\lfloor s/2\rfloor-1}(v)$.
Hence, each vertex in~$T$ has distance at most~$\lfloor s/2\rfloor$ to vertex~$v$.
Thus,~$T$ is an~$s$-club.
It remains to show that each vertex of~$T$ is in a triangle.
Observe  that for each vertex~$w\in N_{\lfloor s/2\rfloor-2}[v]$ we have~$N(w)\subseteq N_{\lfloor s/2\rfloor-1}[v]$.
Recall that since Reduction Rule~\ref{rr-1-vertex-triangle-s-club-remove-vertices-in-no-triangle} is applied, each vertex in~$G$ is contained in a triangle.
Thus, each vertex of~$N_{\lfloor s/2\rfloor-2}[v]$ is contained in a triangle in~$T$.
Furthermore, all vertices in~$N_{\lfloor s/2\rfloor-1}(v)\cup (T\setminus N_{\lfloor s/2\rfloor-1}[v])$ are in a triangle because of the~$T$-expansion. 
Since~$|T|\ge |N_{\lfloor s/2\rfloor-1}[v]|\ge k$, the statement follows.\qed
\end{proof}

Next, we show that Lemma~\ref{lem-1-vertex-triangle-s-club-neighborhood} implies the existence of a Turing kernel for~$s\ge 4$.
We do this by showing that~$N_s[v]$ is bounded for every vertex in the graph.
This in turn implies that the problem is fixed-parameter tractable. 
It is sufficient to bound the size of~$N_s[v]$ for each~$v\in V(G)$ since we then can query the oracle for an $s$-club of size~$k$.

\begin{theorem}
\label{thm-1-vertex-triangle-s-club-fpt}
\textsc{Vertex Triangle~$s$-Club} for~$\ell=1$ admits a~$k^4$-vertex Turing kernel if~$s=4$ or~$s=7$, a~$k^5$-vertex Turing kernel if~$s=5$, and a~$k^3$-vertex Turing kernel if~$s=6$ or~$s\ge 8$.
\end{theorem}
\begin{proof}
First, we apply Reduction Rule~\ref{rr-1-vertex-triangle-s-club-remove-vertices-in-no-triangle}.
Because of Lemma~\ref{lem-1-vertex-triangle-s-club-neighborhood} we conclude that~$(G,k)$ is a trivial yes-instance if~$|N_{\lfloor s/2\rfloor-1}[v]|\ge k$ for any vertex~$v\in V(G)$.
Thus, in the following we can assume that~$|N_{\lfloor s/2\rfloor-1}[v]|<k$ for each vertex~$v\in V(G)$.
We use this fact to bound the size of~$N_s[v]$ in non-trivial instances.

In this case we have~$\lfloor s/2\rfloor-1=1$.
Hence, for each~$s\ge 4$ from~$|N_{\lfloor s/2\rfloor-1}[v]|<k$ we obtain that the size of the neighborhood of each vertex is bounded.
Thus, we obtain a~$k^4$-vertex Turing kernel for~$s=4$ and a~$k^5$-vertex Turing kernel for~$s=5$.
Furthermore, if~$s=7$ we have~$\lfloor s/2\rfloor-1=2$.
Thus, we obtain a~$k^4$-vertex Turing kernel for~$s=7$ since~$N_7[v]\subseteq N_8[v]=N_2[N_2[N_2[N_2[v]]]]$.

If~$s=6$ or~$s\ge 8$, then~$\lfloor s/2\rfloor-1\ge \lceil s/3\rceil$.
Observe that~$N_s[v]$ is contained in~$N_{\lceil s/3\rceil}[N_{\lceil s/3\rceil}[N_{\lceil s/3\rceil}[v]]]$.
Thus, we obtain a~$k^3$-vertex Turing kernel for~$s=6$ or~$s\ge 8$.\qed
\end{proof}

Note that~$s\ge4$ is necessary to ensure~$\lfloor s/2\rfloor-1\ge 1$.
In our arguments to obtain a Turing kernel~$\ell=1$ is necessary for the following reason:
if~$\ell\ge 2$, then the remaining vertices of the other triangles of a vertex in the~$T$-expansion may be contained in~$N_{\lfloor s/2\rfloor+1}$ and, thus,  adding them will not necessarily give an~$s$-club.
Also note that using~$N_t[v]$ for some~$t<\lfloor s/2\rfloor-1$ does not help:
The remaining vertices of the other triangles of a vertex in the~$T$-expansion may be contained in in~$N_{t+1}$. 
But now, another~$T$-expansion for the vertices in~$N_{t+1}$ is necessary.
This may lead to a cascade of~$T$-expansions where eventually, we add vertices with distance at least~$s+1$ to~$v$.
Thus, the constructed set is no $s$-club anymore.

\subsection{Parameterized Hardness}

In the following, we prove W[1]-hardness for \textsc{Vertex Triangle~$s$-Club} parameterized by the solution size~$k$ for all cases not covered by Theorem~\ref{thm-1-vertex-triangle-s-club-fpt}, that is,~$\ell\ge 2$ and~$s\ge 2$, and also for~$\ell=1$ and $s\in\{2,3\}$. 

\begin{theorem}
\label{thm-w1-hardness-vertex-variant}
	\textsc{Vertex Triangle~$s$-Club} is W[1]-hard for parameter~$k$ if~$\ell\ge 2$, and if~$\ell=1$ and~$s\in\{2,3\}$.
\end{theorem}

For some combinations of~$s$ and~$\ell$ we provide hardness for restricted input graphs.
More precisely, we prove that \textsc{Vertex Triangle~$s$-Club} is W[1]-hard even if each vertex~$v\in V(G)$ is contained in \emph{exactly}~$\ell$ triangles in the input graph.
In other words, the hardness does not depend on the fact that we could choose different triangles.
We provide this hardness for the case~$s\ge 3$ and arbitrary~$\ell$ and also for~$s=2$ when~$\ell=\binom{c-1}{2}$ for some integer~$c$.

We prove the theorem by considering four subcases. 
The proofs for the four cases all use a reduction from the W[1]-hard \textsc{Clique} problem.
In these constructions, each vertex~$v$ of the \textsc{Clique} instance is replaced by a vertex gadget~$T^v$ such that every vertex-$\ell$-triangle~$s$-club~$S$  either contains~$T^v$ completely or contains no vertex of~$T^v$.
This property is obtained since each vertex in~$T^v$ will be in exactly $\ell$~triangles and each of these triangles is within~$T^v$.
The idea is that if~$uv\notin E(G)$ then there exists a vertex~$x\in T^u$ and a vertex~$y\in T^v$ such that~$\dist(x,y)\ge s+1$.

\paragraph{Vertex Triangle~$2$-Club.} 
First, we handle the case~$s=2$.

\begin{const}
\label{const-l-triangle-2-club}
Let~$(G,k)$ be an instance of \textsc{Clique} and let~$c$ be the smallest number such that~$\binom{c-1}{2}\ge\ell$.
We construct an instance~$(G',c(k+1))$ of \textsc{Vertex Triangle~$2$-Club} as follows.
\begin{itemize}
\item For each vertex~$v\in V(G)$, we add a clique~$T^v\coloneqq \{x_1^v, \ldots , x_c^v\}$ of size~$c$ to~$G'$. 

\item For each edge~$vw\in E(G)$, we connect the cliques~$T^v$ and~$T^w$ by adding the edge~$x^v_{2i-1}x^w_{2i}$ and~$x^w_{2i-1}x^v_{2i}$ to~$G'$ for each~$i\in[\lfloor c/2\rfloor]$.

\item Furthermore, we add a clique~$Y\coloneqq \{y_1, \ldots, y_c\}$ of size~$c$ to~$G'$.

\item We also add, for each~$i\in[c]$ and each~$v\in V(G)$, the edge~$x^v_iy_i$ to~$G'$.
\end{itemize}
\end{const}

Note that the clique size~$c$ ensures that each vertex~$x\in V(G')$ is contained in at least~$\binom{c-1}{2}\ge\ell$ triangles in~$G'$.
Furthermore, note that the clique~$Y$ is only necessary when~$c$ is odd to ensure that the vertices~$x^v_c$ and~$x^w_c$ have a common neighbor.
We add the clique~$Y$ in both cases to unify the construction and the correctness proof.
Next, we show that for each vertex gadget~$T^v$ the intersection with each vertex-$\ell$-triangle~$2$-club is either empty or~$T^v$.

\begin{lemma}
\label{lem-l-triangle-2-club-property-empty-or-complete} Let~$S$ be a vertex-$\ell$-triangle~$2$-club in~$G'$. Then,
\begin{itemize}
\item[a)] $S\cap T^v\neq\emptyset\Leftrightarrow T^v\subseteq S$, and
\item[b)]  $S'\coloneqq S\cup Y$ is also a vertex-$\ell$-triangle~$2$-club in~$G'$.
\end{itemize}
\end{lemma} 
\begin{proof}
First, we show statement~$a)$. 
Assume that for a vertex~$z\in T^v$ for some~$v\in V(G)$ we have~$z\in S$ for some vertex-$\ell$-triangle~$2$-club~$S$.
Note that~$T^v$ contains all vertices which form a triangle with vertex~$z$. 
Since~$c$ is minimal such that~$\binom{c-1}{2}\ge\ell$ and since~$T^v$ is a clique, we conclude that~$T^v\subseteq S$ to fulfill the property that vertex~$z$ is contained in at least~$\ell$ triangles in~$G[S]$.
Thus,~$T^v\subseteq S$.

Second, we show statement~$b)$.
Since each vertex~$y\in Y$ forms only triangles with vertices in~$Y$ and~$Y$ has size~$\binom{c-1}{2}$, we conclude that~$Y\subseteq S^*\Leftrightarrow S^*\cap Y\ne\emptyset$ for each vertex-$\ell$-triangle~$2$-club~$S^*$. 
In the following, let~$S$ be a vertex-$\ell$-triangle~$2$-club such that~$Y\cap S=\emptyset$.
From statement~$a)$ we conclude that~$S\coloneqq \bigcup_{v\in P}T^v$ for some set~$P\subseteq V(G)$.
Next, we show that~$S'\coloneqq S\cup Y$ is also a vertex-$\ell$-triangle~$2$-club.
Since each vertex is contained in a clique of size~$c$ in~$S$, it is in sufficiently many triangles.
Thus, it remains to prove that~$S'$ is a~$2$-club.
Hence, consider some vertex~$x^v_i$ and some vertex~$y_j$ for some~$i,j\in[c]$ and~$v\in P$. 
Then,~$y_i$ is a common neighbor of~$x^v_i$ and~$y_j$.
Hence,~$S'$ is a vertex-$\ell$-triangle~$2$-club and thus~$b)$ holds.\qed
\end{proof}

Now, we prove the correctness of Construction~\ref{const-l-triangle-2-club}.

\begin{lemma}
\label{lem-l-triangle-2-club-hardness}
For each~$\ell\in\mathds{N}$, the \textsc{Vertex Triangle~$2$-Club} problem parameterized by~$k$ is W[1]-hard.
\end{lemma}

\begin{proof}
We prove that~$G$ contains a clique of size~$k$ if and only if~$G'$ contains a vertex-$\ell$-triangle~$2$-club of size at least~$c(k+1)$.

Let~$C$ be a clique of size at least~$k$ in~$G$.
We argue that~$S\coloneqq Y\cup\bigcup_{v\in C}T^v$ is a vertex-$\ell$-triangle~$2$-club of size~$c(k+1)$ in~$G'$.
Note that for each vertex~$v\in T^v$ we have~$|T^v|=c$.
Since~$T^v$ is a clique, we conclude that each vertex in~$T^v$ is contained in exactly~$\binom{c-1}{2}\ge\ell$ triangles.
The same is true for each vertex in~$Y$.
Hence, each vertex in~$S$ is contained in at least~$\ell$ triangles.
Thus, it remains to show that~$S$ is a~$2$-club.
Consider the vertices~$x^v_i$ and~$x^w_j$ for~$v,w\in C$,~$i\in[c-1]$, and~$j\in[c]$. 
If~$i$ is odd, then~$x^w_{i+1}\in N(x^v_i)\cap N(x^w_j)$.
Otherwise, if~$i$ is even,~$x^w_{i-1}\in N(x^v_i)\cap N(x^w_j)$.
In both cases, we obtain~$\dist(x^v_i,x^w_j)\le 2$.
Next, consider two vertices~$x^v_c$ and~$x^w_c$ in~$S$.
Observe that~$y_c\in N(x^v_c)\cap N(x^w_c)$.
Since~$Y$ is a clique, it remains to consider vertices~$x^v_i$ and~$y_j$ in~$S$ for~$i\in[c]$ and~$j\in[c]$. 
Observe that~$x^v_j\in N(y_j)\cap N[x^v_i]$.
Thus,~$G'$ contains a vertex-$\ell$-triangle~$2$-club of size at least~$c(k+1)$.

Conversely, suppose that~$G'$ contains a vertex-$\ell$-triangle~$2$-club~$S$ of size at least~$c(k+1)$.
By Lemma~\ref{lem-l-triangle-2-club-property-empty-or-complete}, we can assume that~$Y\subseteq S$ and for each vertex gadget~$T^v\in G'$ we either have~$T^v\subseteq S$ or~$T^v\cap S=\emptyset$.
Hence,~$S$ contains at least~$k$ cliques of the form~$T^v$.
Assume towards a contradiction that~$S$ contains two cliques~$T^v$ and~$T^w$ such that~$vw\notin E(G)$ and consider the two vertices~$x^v_1\in T^v$ and~$x^w_2\in T^w$.
Note that these vertices always exist since~$c\ge 3$.
Observe that~$N[x^v_1]=T^v\cup\{x^u_2\mid uv\in E(G)\}\cup\{y_1\}$ and~$N[x^w_2]=T^w\cup\{x^u_1\mid uw\in E(G)\}\cup\{y_2\}$.
Thus,~$\dist(x^v_1,x^w_2)\ge 3$, a contradiction.
Hence, for each two distinct vertex gadgets~$T^v$ and~$T^w$ that are contained in~$S$, we observe that~$vw\in E(G)$. 
Consequently, the set~$\{v \mid T^v\subseteq S\}$ is a clique of size at least~$k$ in~$G$.\qed
\end{proof}

If~$\ell=\binom{c-1}{2}$ for some integer~$c$, then Lemma~\ref{lem-l-triangle-2-club-hardness} also holds for the restriction that each vertex is contained in exactly~$\ell$ triangles in the input graph~$G'$.

\paragraph{Vertex Triangle~$s$-Club for~$s=3$ and for~$\ell\ge 2$ and~$s\ge 4$.}
Now, we provide hardness for the remaining cases.
We consider three subcases.
Case~$1$ deals with odd~$s$.
Case~$2$ covers the case that~$s$ is even and~$\ell\ge 3$.
Case~$3$ deals with the case that~$s$ is even and~$\ell=2$.
All three cases use the same vertex gadget. 
Only the edges between these gadgets, called \emph{connector edges}, differ.
The idea is to construct the vertex gadgets~$T^v$ in such a way that there are pairs of vertices  in~$T^v$ of distance~$2s^*$ which is almost~$s$.
Thus, in a vertex-$\ell$-triangle~$s$-club, the distance between two different vertex gadgets must be small.
The first part of the following construction describes the vertex gadget which is used in all three cases.
For an illustration of this construction see Fig.~\ref{fig-visulation-vertex-triangle}.

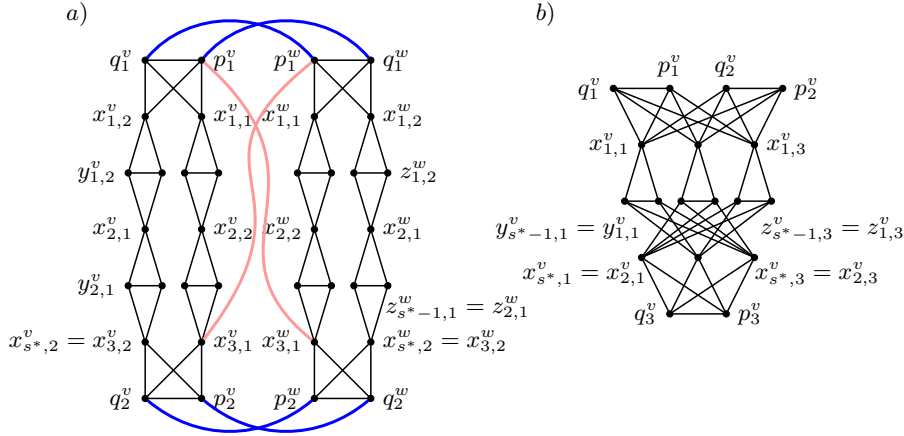
\begin{figure}[t]
\begin{tikzpicture}[scale=0.75]

\node[label=:{$a)$}](p) at (-1.2, 6.3) {};

\node[label=left:{\footnotesize $q^v_2$}](pv2) at (0, 0) [shape = circle, draw, fill=black, scale=0.07ex]{};
\node[label=right:{\footnotesize $p^v_2$}](qv2) at (1, 0) [shape = circle, draw, fill=black, scale=0.07ex]{};
\node[label=left:{\footnotesize $x^v_{s^*,2}=x^v_{3,2}$}](xv31) at (0, 1) [shape = circle, draw, fill=black, scale=0.07ex]{};
\node[label=right:{\footnotesize $x^v_{3,1}$}](xv32) at (1, 1) [shape = circle, draw, fill=black, scale=0.07ex]{};
\node[label=left:{\footnotesize $y^v_{2,1}$}](yv21) at (-0.3, 2) [shape = circle, draw, fill=black, scale=0.07ex]{};
\node(zv21) at (0.3, 2) [shape = circle, draw, fill=black, scale=0.07ex]{};
\node(yv22) at (0.7, 2) [shape = circle, draw, fill=black, scale=0.07ex]{};
\node(zv22) at (1.3, 2) [shape = circle, draw, fill=black, scale=0.07ex]{};
\node[label=left:{\footnotesize $x^v_{2,1}$}](xv21) at (0, 3) [shape = circle, draw, fill=black, scale=0.07ex]{};
\node[label=right:{\footnotesize $x^v_{2,2}$}](xv22) at (1, 3) [shape = circle, draw, fill=black, scale=0.07ex]{};
\node[label=left:{\footnotesize $y^v_{1,2}$}](yv11) at (-0.3, 4) [shape = circle, draw, fill=black, scale=0.07ex]{};
\node(zv11) at (0.3, 4) [shape = circle, draw, fill=black, scale=0.07ex]{};
\node(yv12) at (0.7, 4) [shape = circle, draw, fill=black, scale=0.07ex]{};
\node(zv12) at (1.3, 4) [shape = circle, draw, fill=black, scale=0.07ex]{};
\node[label=left:{\footnotesize $x^v_{1,2}$}](xv11) at (0, 5) [shape = circle, draw, fill=black, scale=0.07ex]{};
\node[label=right:{\footnotesize $x^v_{1,1}$}](xv12) at (1, 5) [shape = circle, draw, fill=black, scale=0.07ex]{};
\node[label=left:{\footnotesize $q^v_1$}](pv1) at (0, 6) [shape = circle, draw, fill=black, scale=0.07ex]{};
\node[label=right:{\footnotesize $p^v_1$}](qv1) at (1, 6) [shape = circle, draw, fill=black, scale=0.07ex]{};

\path [-,line width=0.2mm] (pv2) edge (qv2);
\path [-,line width=0.2mm] (pv2) edge (xv31);
\path [-,line width=0.2mm] (pv2) edge (xv32);
\path [-,line width=0.2mm] (qv2) edge (xv31);
\path [-,line width=0.2mm] (qv2) edge (xv32);
\path [-,line width=0.2mm] (pv1) edge (qv1);
\path [-,line width=0.2mm] (pv1) edge (xv11);
\path [-,line width=0.2mm] (pv1) edge (xv12);
\path [-,line width=0.2mm] (qv1) edge (xv11);
\path [-,line width=0.2mm] (qv1) edge (xv12);

\path [-,line width=0.2mm] (zv11) edge (xv11);
\path [-,line width=0.2mm] (yv11) edge (xv11);
\path [-,line width=0.2mm] (zv12) edge (xv12);
\path [-,line width=0.2mm] (yv12) edge (xv12);
\path [-,line width=0.2mm] (zv11) edge (yv11);
\path [-,line width=0.2mm] (yv12) edge (zv12);
\path [-,line width=0.2mm] (zv11) edge (xv21);
\path [-,line width=0.2mm] (yv11) edge (xv21);
\path [-,line width=0.2mm] (zv12) edge (xv22);
\path [-,line width=0.2mm] (yv12) edge (xv22);

\path [-,line width=0.2mm] (zv21) edge (xv21);
\path [-,line width=0.2mm] (yv21) edge (xv21);
\path [-,line width=0.2mm] (zv22) edge (xv22);
\path [-,line width=0.2mm] (yv22) edge (xv22);
\path [-,line width=0.2mm] (zv21) edge (yv21);
\path [-,line width=0.2mm] (yv22) edge (zv22);
\path [-,line width=0.2mm] (zv21) edge (xv31);
\path [-,line width=0.2mm] (yv21) edge (xv31);
\path [-,line width=0.2mm] (zv22) edge (xv32);
\path [-,line width=0.2mm] (yv22) edge (xv32);

\node[label=left:{\footnotesize $p^w_2$}](pw2) at (3, 0) [shape = circle, draw, fill=black, scale=0.07ex]{};
\node[label=right:{\footnotesize $q^w_2$}](qw2) at (4, 0) [shape = circle, draw, fill=black, scale=0.07ex]{};
\node[label=left:{\footnotesize $x^w_{3,1}$}](xw31) at (3, 1) [shape = circle, draw, fill=black, scale=0.07ex]{};
\node[label=right:{\footnotesize $x^w_{s^*,2}=x^w_{3,2}$}](xw32) at (4, 1) [shape = circle, draw, fill=black, scale=0.07ex]{};
\node(yw21) at (2.7, 2) [shape = circle, draw, fill=black, scale=0.07ex]{};
\node(zw21) at (3.3, 2) [shape = circle, draw, fill=black, scale=0.07ex]{};
\node(yw22) at (3.7, 2) [shape = circle, draw, fill=black, scale=0.07ex]{};
\node(zw22) at (4.3, 2) [shape = circle, draw, fill=black, scale=0.07ex]{};
\node[label=right:{\footnotesize $z^w_{s^*-1,1}=z^w_{2,1}$}](zw22a) at (3.95, 1.6) {};
\node[label=left:{\footnotesize $x^w_{2,2}$}](xw21) at (3, 3) [shape = circle, draw, fill=black, scale=0.07ex]{};
\node[label=right:{\footnotesize $x^w_{2,1}$}](xw22) at (4, 3) [shape = circle, draw, fill=black, scale=0.07ex]{};
\node(yw11) at (2.7, 4) [shape = circle, draw, fill=black, scale=0.07ex]{};
\node(zw11) at (3.3, 4) [shape = circle, draw, fill=black, scale=0.07ex]{};
\node(yw12) at (3.7, 4) [shape = circle, draw, fill=black, scale=0.07ex]{};
\node[label=right:{\footnotesize $z^w_{1,2}$}](zw12) at (4.3, 4) [shape = circle, draw, fill=black, scale=0.07ex]{};
\node[label=left:{\footnotesize $x^w_{1,1}$}](xw11) at (3, 5) [shape = circle, draw, fill=black, scale=0.07ex]{};
\node[label=right:{\footnotesize $x^w_{1,2}$}](xw12) at (4, 5) [shape = circle, draw, fill=black, scale=0.07ex]{};
\node[label=left:{\footnotesize $p^w_1$}](pw1) at (3, 6) [shape = circle, draw, fill=black, scale=0.07ex]{};
\node[label=right:{\footnotesize $q^w_1$}](qw1) at (4, 6) [shape = circle, draw, fill=black, scale=0.07ex]{};

\path [-,line width=0.2mm] (pw2) edge (qw2);
\path [-,line width=0.2mm] (pw2) edge (xw31);
\path [-,line width=0.2mm] (pw2) edge (xw32);
\path [-,line width=0.2mm] (qw2) edge (xw31);
\path [-,line width=0.2mm] (qw2) edge (xw32);
\path [-,line width=0.2mm] (pw1) edge (qw1);
\path [-,line width=0.2mm] (pw1) edge (xw11);
\path [-,line width=0.2mm] (pw1) edge (xw12);
\path [-,line width=0.2mm] (qw1) edge (xw11);
\path [-,line width=0.2mm] (qw1) edge (xw12);

\path [-,line width=0.2mm] (zw11) edge (xw11);
\path [-,line width=0.2mm] (yw11) edge (xw11);
\path [-,line width=0.2mm] (zw12) edge (xw12);
\path [-,line width=0.2mm] (yw12) edge (xw12);
\path [-,line width=0.2mm] (zw11) edge (yw11);
\path [-,line width=0.2mm] (yw12) edge (zw12);
\path [-,line width=0.2mm] (zw11) edge (xw21);
\path [-,line width=0.2mm] (yw11) edge (xw21);
\path [-,line width=0.2mm] (zw12) edge (xw22);
\path [-,line width=0.2mm] (yw12) edge (xw22);

\path [-,line width=0.2mm] (zw21) edge (xw21);
\path [-,line width=0.2mm] (yw21) edge (xw21);
\path [-,line width=0.2mm] (zw22) edge (xw22);
\path [-,line width=0.2mm] (yw22) edge (xw22);
\path [-,line width=0.2mm] (zw21) edge (yw21);
\path [-,line width=0.2mm] (yw22) edge (zw22);
\path [-,line width=0.2mm] (zw21) edge (xw31);
\path [-,line width=0.2mm] (yw21) edge (xw31);
\path [-,line width=0.2mm] (zw22) edge (xw32);
\path [-,line width=0.2mm] (yw22) edge (xw32);

\path [-,myblue,line width=0.4mm,bend left=50] (pv1) edge (pw1);
\path [-,myblue,line width=0.4mm,bend left=50] (qv1) edge (qw1);
\path [-,myblue,line width=0.4mm,bend right=40] (pv2) edge (pw2);
\path [-,myblue,line width=0.4mm,bend right=40] (qv2) edge (qw2);
\draw [-,myred,line width=0.4mm] (qv1) .. controls (3.5,3.8) and (0.9,2.7) .. (xw31);
\draw [-,myred,line width=0.4mm] (pw1) .. controls (0.5,4) and (3.2,3.8) .. (xv32);


\node[label=:{$b)$}](ap) at (7.1, 6.3) {};

\node[label=left:{\footnotesize $q^v_3$}](apv3) at (9.3, 1.5) [shape = circle, draw, fill=black, scale=0.07ex]{};
\node[label=right:{\footnotesize $p^v_3$}](aqv3) at (10.3, 1.5) [shape = circle, draw, fill=black, scale=0.07ex]{};
\node(axv21) at (8.8, 2.5) [shape = circle, draw, fill=black, scale=0.07ex]{};
\node[label=left:{\footnotesize $x^v_{s^*,1}=x^v_{2,1}$}](axv21a) at (9.2, 2.2) {};
\node(axv22) at (9.8, 2.5) [shape = circle, draw, fill=black, scale=0.07ex]{};
\node(axv23) at (10.8, 2.5) [shape = circle, draw, fill=black, scale=0.07ex]{};
\node[label=right:{\footnotesize $x^v_{s^*,3}=x^v_{2,3}$}](axv23s) at (10.5, 2.2) {};

\node(ayv11) at (8.5, 3.5) [shape = circle, draw, fill=black, scale=0.07ex]{};
\node[label=left:{\footnotesize $y^v_{s^*-1,1}=y^v_{1,1}$}](ayv11s) at (9.1, 3.0) {};
\node(azv11) at (9.1, 3.5) [shape = circle, draw, fill=black, scale=0.07ex]{};
\node(ayv12) at (9.5, 3.5) [shape = circle, draw, fill=black, scale=0.07ex]{};
\node(azv12) at (10.1, 3.5) [shape = circle, draw, fill=black, scale=0.07ex]{};
\node(ayv13) at (10.5, 3.5) [shape = circle, draw, fill=black, scale=0.07ex]{};
\node(azv13) at (11.1, 3.5) [shape = circle, draw, fill=black, scale=0.07ex]{};
\node[label=right:{\footnotesize $z^v_{s^*-1,3}=z^v_{1,3}$}](azv13s) at (10.6, 3.0) {};

\node[label=right:{\footnotesize $x^v_{1,3}$}](axv13) at (10.8, 4.5) [shape = circle, draw, fill=black, scale=0.07ex]{};
\node(axv11) at (9.8, 4.5) [shape = circle, draw, fill=black, scale=0.07ex]{};
\node[label=left:{\footnotesize $x^v_{1,1}$}](axv12) at (8.8, 4.5) [shape = circle, draw, fill=black, scale=0.07ex]{};
\node[label=left:{\footnotesize $q^v_1$}](apv1) at (8.3, 5.5) [shape = circle, draw, fill=black, scale=0.07ex]{};
\node[label=above:{\footnotesize $p^v_1$}](aqv1) at (9.3, 5.5) [shape = circle, draw, fill=black, scale=0.07ex]{};
\node[label=above:{\footnotesize $q^v_2$}](apv2) at (10.3, 5.5) [shape = circle, draw, fill=black, scale=0.07ex]{};
\node[label=right:{\footnotesize $p^v_2$}](aqv2) at (11.3, 5.5) [shape = circle, draw, fill=black, scale=0.07ex]{};

\path [-,line width=0.2mm] (apv1) edge (aqv1);
\path [-,line width=0.2mm] (apv2) edge (aqv2);
\path [-,line width=0.2mm] (apv3) edge (aqv3);
\path [-,line width=0.2mm] (apv3) edge (axv21);
\path [-,line width=0.2mm] (apv3) edge (axv22);
\path [-,line width=0.2mm] (apv3) edge (axv23);
\path [-,line width=0.2mm] (aqv3) edge (axv21);
\path [-,line width=0.2mm] (aqv3) edge (axv22);
\path [-,line width=0.2mm] (aqv3) edge (axv23);

\path [-,line width=0.2mm] (apv1) edge (axv11);
\path [-,line width=0.2mm] (apv1) edge (axv12);
\path [-,line width=0.2mm] (apv1) edge (axv13);
\path [-,line width=0.2mm] (aqv1) edge (axv11);
\path [-,line width=0.2mm] (aqv1) edge (axv12);
\path [-,line width=0.2mm] (aqv1) edge (axv13);

\path [-,line width=0.2mm] (apv2) edge (axv11);
\path [-,line width=0.2mm] (apv2) edge (axv12);
\path [-,line width=0.2mm] (apv2) edge (axv13);
\path [-,line width=0.2mm] (aqv2) edge (axv11);
\path [-,line width=0.2mm] (aqv2) edge (axv12);
\path [-,line width=0.2mm] (aqv2) edge (axv13);

\path [-,line width=0.2mm] (ayv11) edge (axv12);
\path [-,line width=0.2mm] (azv11) edge (axv12);
\path [-,line width=0.2mm] (azv11) edge (ayv11);
\path [-,line width=0.2mm] (azv12) edge (ayv12);
\path [-,line width=0.2mm] (azv13) edge (ayv13);
\path [-,line width=0.2mm] (ayv12) edge (axv11);
\path [-,line width=0.2mm] (azv12) edge (axv11);
\path [-,line width=0.2mm] (ayv13) edge (axv13);
\path [-,line width=0.2mm] (azv13) edge (axv13);

\path [-,line width=0.2mm] (ayv11) edge (axv22);
\path [-,line width=0.2mm] (azv11) edge (axv22);
\path [-,line width=0.2mm] (ayv11) edge (axv23);
\path [-,line width=0.2mm] (azv11) edge (axv23);
\path [-,line width=0.2mm] (ayv12) edge (axv21);
\path [-,line width=0.2mm] (azv12) edge (axv21);
\path [-,line width=0.2mm] (ayv12) edge (axv23);
\path [-,line width=0.2mm] (azv12) edge (axv23);
\path [-,line width=0.2mm] (ayv13) edge (axv22);
\path [-,line width=0.2mm] (azv13) edge (axv22);
\path [-,line width=0.2mm] (ayv13) edge (axv21);
\path [-,line width=0.2mm] (azv13) edge (axv21);
\end{tikzpicture}
\caption{$a)$ The vertex gadgets~$T^v$ and~$T^w$ for~$s\in\{7,8\}$ and~$\ell=2$. 
The blue lines are only added if~$s$ is odd and the red lines are only added if~$s$ is even.
$b)$ The vertex gadget~$T^v$ for~$s\in\{5,6\}$ and~$\ell=3$.}
\label{fig-visulation-vertex-triangle}
\end{figure}

\begin{const}
\label{const-l-triangle-3-or-4-club}
We set~$s^*\coloneqq \lfloor (s-1)/2\rfloor$.
Let~$(G,k)$ be an instance of \textsc{Clique}. 
We construct an equivalent instance~$(G',3\ell ks^*)$ of \textsc{Vertex Triangle~$s$-Club}.
Recall that~$\ell\ge 1$ and~$s=3$, or~$\ell\ge 2$ and~$s\ge 4$.
For each vertex~$v\in V(G)$ we construct a vertex gadget~$T^v$.
This vertex gadget is used for each reduction of the three cases.
\begin{itemize}
\item We add the vertices~$p^v_i$, and~$q^v_i$, and an edge~$p^v_iq^v_i$ for each~$i\in[\ell]$ to~$G'$.
\item We add a vertex~$x^v_{j,i}$ for each~$i\in[\ell]$ and each~$j\in[s^*]$ to~$G'$.
\item We add an edge~$y^v_{t,i}z^v_{t,i}$ for each~$i\in[\ell]$ and each~$t\in[s^*-1]$ to~$G'$.
Note that these vertices only exist, if~$s\ge 5$.
\end{itemize}

The vertices~$x^v_{j,i}$,~$y^v_{t,i}$, and~$z^v_{t,i}$ are referred to as the \emph{cascading vertices}. 
They are used to ensure that all vertices in~$T^v$ are in exactly~$\ell$ triangles and that there are vertex pairs of distance~$2s^*$ within~$T^v$.
Note that since~$s\ge 3$ we create at least~$\ell$ many~$x$-vertices.
We connect these vertices as follows:
\begin{itemize}
\item We add the edges~$p^v_ix^v_{1,j}$ and~$q^v_ix^v_{1,j}$ for each~$i\in[\ell-1]$ and each~$j\in[\ell]$ to~$G'$.
\item We add the edges~$p^v_\ell x^v_{s^*,j}$ and~$q^v_\ell x^v_{s^*,j}$ for each~$j\in[\ell]$ to~$G'$.
\item We add the edges~$y^v_{t,i}x^v_{t,i}$ and~$z^v_{t,i}x^v_{t,i}$ for each~$i\in[\ell]$, and each~$t\in[s^*-1]$ to~$G'$.
\item We add the edges~$y^v_{t,i}x^v_{t+1,j}$ and~$z^v_{t,i}x^v_{t+1,j}$ for each~$i\in[\ell]$, each~$j\in[\ell]\setminus\{i\}$, and each~$t\in[s^*-1]$ to~$G'$.
\end{itemize}

Note that if~$s=3$ or~$s=4$, the graph~$G'$ is a non-induced biclique where one partite set consists of the vertices~$\{x^v_{1,j}\mid j\in[\ell]\}$ and the other partite set consists of the vertices~$\{\{p^v_i,q^v_i\}\mid i\in[\ell]\}$.
Furthermore, the additional edges are~$p^v_iq^v_i$ for each~$i\in[\ell]$.
Also, observe that each vertex gadget~$T^v$ consists of exactly~$3\ell s^*$ vertices.

From now on, the construction for the three cases differs.
We now connect these vertex gadgets by introducing the \emph{connector edges}:
For each edge~$vw\in E(G)$ we add edges between the vertex gadgets~$T^v$ and~$T^w$ of the corresponding vertices. 
We distinguish the three cases. 
\begin{enumerate}[label=\textbf{Case \Roman*:},leftmargin=*]
\item \textbf{$s$ is odd.} We add the edges~$p^v_iq^w_i$ and~$q^v_ip^w_i$ to~$G'$ for each~$i\in[\ell]$ to~$G'$.
\item \textbf{$s$ is even and~$\ell\ge 3$.} We add the edges~$p^v_1q^w_1$,~$q^v_1p^w_1$,~$p^v_\ell q^w_\ell$, and~$q^v_\ell p^w_\ell$ to~$G'$.
\item \textbf{$s$ is even and~$\ell=2$.} We add the edges~$p^v_1x^w_{s^*,1}$, and~$p^w_1x^v_{s^*,1}$ to~$G'$.
\end{enumerate}
\end{const}

We make the following observation about the connector edges between different vertex gadgets:
If~$s$ is odd (Case~I), or if~$s$ is even and~$\ell\ge 3$ (Case~II), we  have~$N(p^v_i)\setminus T^v\subseteq\{q^w_i\mid vw\in E(G)\}$ and also~$N(q^v_i)\setminus T^v\subseteq\{p^w_i\mid vw\in E(G)\}$ for each~$v\in V(G)$ and each~$i\in[\ell]$.
Otherwise, if~$s$ is even and~$\ell=2$ (Case~III), observe that we have~$N(p^v_1)\setminus T^v=\{x^w_{s*,1}\mid vw\in E(G)\}$ and also~$N(x^v_{s^*,1})\setminus T^v=\{p^w_{1}\mid vw\in E(G)\}$ for each~$v\in V(G)$.
Since only the connector edges have endpoints in different vertex gadgets, we thus observe the following.

\begin{observation}
\label{obs-l-vertex-triangle-s-club-triangles-in-vertex-gadgets}
All~$3$ endpoints of each triangle in~$G'$ are contained in exactly one vertex gadget.
\end{observation}

Next, we show that each vertex in a vertex gadget~$T^v$ is contained in exactly~$\ell$ triangles.
Together with Observation~\ref{obs-l-vertex-triangle-s-club-triangles-in-vertex-gadgets} this implies that each vertex in~$G'$ is contained in exactly~$\ell$ triangles.

\begin{lemma}
\label{lem-l-vertex-triangle-s-club-each-vertex-in-l-triangles}
Let~$T^v$ be a vertex gadget. 
Each vertex in~$T^v$ is contained in exactly~$\ell$ triangles.
\end{lemma}

\begin{proof}
We make a case distinction, that is, for each vertex~$a\in T^v$ we present exactly~$\ell$ triangles containing vertex~$a$.
We distinguish the different vertices of a vertex gadget.

\textbf{Case 1:} Consider vertex~$p^v_d$ for some~$d\in[\ell-1]$.
Because of Observation~\ref{obs-l-vertex-triangle-s-club-triangles-in-vertex-gadgets} we only have to consider the neighbors of~$p^v_d$ in~$T^v$.
By construction we have~$N(p^v_d)\cap T^v=\{q^v_d\}\cup\{x^v_{1,j}\mid j\in[\ell]\}$.
Observe that the only edges within~$N(p^v_d)\cap T^v$ are the edges~$q^v_dx^v_{1,j}$ for each~$j\in[\ell]$.
Thus,~$p^v_d$ is contained in the~$\ell$ triangles~$\{\{p^v_d,q^v_d,x^v_{1,i}\}\mid i\in[\ell]\}$. 
By similar arguments the same is true for vertex~$q^v_d$.

\textbf{Case 2:} Consider vertex~$p^v_\ell$.
By construction we have~$N(p^v_\ell)= \{q^v_\ell\}\cup\{x^v_{s^*,i}\mid i\in[\ell]\}\cup\{q^w_d\mid vw\in E(G)\}$ if~$s$ is odd or~$s$ is even and~$\ell\ge 3$.
If~$s$ is even and~$\ell=2$ we have~$N(p^v_\ell)=\{q^v_\ell\}\cup\{x^v_{s^*,i}\mid i\in[\ell]\}$.
Observe that the only edges within~$N(p^v_\ell)$ are the edges~$q^v_\ell x^v_{s^*,i}$ for each~$i\in[\ell]$.
Thus,~$p^v_\ell$ is contained in the~$\ell$ triangles~$\{\{p^v_\ell,q^v_\ell,x^v_{s^*,i}\}\mid i\in[\ell]\}$. 
By similar arguments the same is true for vertex~$q^v_\ell$.

\textbf{Case 3:} Now, consider vertices~$x^v_{1,i}$ and~$x^v_{s^*,i}$ for some~$i\in[\ell]$.
Here we have to distinguish if~$s\in\{3,4\}$ or~$s\ge 5$ since~$y^v_{t,i}$ exists only in the second case.

First, consider the case~$s=3$ or~$s=4$.
Note that we now have~$x^v_{1,i}=x^v_{s^*,i}$.
By construction we have~$N(x^v_{1,i})=\{p^v_j\mid j\in[\ell]\}\cup\{q^v_j\mid j\in[\ell]\}$.
Note that the only edges within~$N(x^v_{1,i})$ have the form~$p^v_jq^v_j$ for each~$j\in[\ell]$.
Thus,~$x^v_{1,i}$ is contained in the~$\ell$ triangles~$\{\{x^v_{1,i},p^v_j,q^v_j\}\mid j\in[\ell]\}$.

Second, consider the case~$s\ge 5$. 
First, we investigate vertex~$x^v_{1,i}$.
By construction we have~$N(x^v_{1,i})=\{p^v_j\mid j\in[\ell-1]\}\cup\{q^v_j\mid j\in[\ell-1]\}\cup\{y^v_{1,i},z^v_{1,i}\}$.
The only edges within~$N(x^v_{1,i})$ are the edge~$p^v_jq^v_j$ for each~$j\in[\ell-1]$ and the edge~$y^v_{1,i}z^v_{1,i}$.
Thus,~$x^v_{1,i}$ is contained in the~$\ell-1$ triangles~$\{\{x^v_{1,i},p^v_j,q^v_j\}\mid j\in[\ell-1]\}$, and in the triangle~$\{x^v_{1,i},y^v_{1,i},z^v_{1,i}\}$.

Now, consider vertex~$x^v_{s^*,i}$.
Because of Observation~\ref{obs-l-vertex-triangle-s-club-triangles-in-vertex-gadgets} we only have to consider the neighbors of~$x^v_{s^*,i}$ in~$T^v$.
By construction we observe that~$N(x^v_{s^*,i})\cap T^v\subseteq \{p^v_\ell,q^v_\ell\}\cup\{y^v_{s^*-1,j},z^v_{s^*-1,j}\mid j\in[\ell]\setminus\{i\}\}$.
Note that the only edges within~$N(x^v_{s^*,i})$ are the edge~$p^v_\ell q^v_\ell$ and the edge~$y^v_{s^*-1,j}z^v_{s^*-1,j}$ for each~$j\in[\ell]\setminus\{i\}$.
Thus,~$x^v_{s^*,i}$ is contained in the triangle~$\{x^v_{s^*,i},p^v_\ell,q^v_\ell\}$ and in the~$\ell-1$ triangles $\{\{x^v_{s^*,i},y^v_{s^*-1,j},z^v_{s^*-1,j}\}\mid j\in[\ell]\setminus\{i\}\}$.

\textbf{Case 4:} Consider vertex~$x^v_{r,i}$ for some~$i\in[\ell]$ and some~$r\in[2,s^*-1]$.
Recall that these vertices only exist if~$s\ge 7$.
By construction~$N(x^v_{r,i})=\{y^v_{r,i},z^v_{r,i}\}\cup\{y^v_{r-1,j},z^v_{r-1,j}\mid j\in[\ell]-1\}$.
The only edges within~$N(x^v_{r,i})$ have the form~$y^v_{r,i}z^v_{r,i}$ and~$y^v_{r-1,j}z^v_{r-1,j}$ for each~$j\in[\ell]\setminus\{i\}$.
Thus,~$x^v_{r,i}$ is contained in the triangle~$\{x^v_{r,i},y^v_{r,i},z^v_{r,i}\}$ and in the~$\ell-1$ triangles~$\{\{x^v_{r,i},y^v_{r-1,j},z^v_{r-1,j}\}\mid j\in[\ell]\setminus\{i\}\}$.

\textbf{Case 5:} Finally, consider vertex~$y^v_{t,i}$ for some~$i\in[\ell]$ and some~$t\in[s^*-1]$.
Recall that these vertices only exist if~$s\ge 5$.
By construction we have~$N(y^v_{t,i})=\{x^v_{t,i},z^v_{t,i}\}\cup\{x^v_{t+1,j}\mid j\in[\ell]\setminus\{i\}\}$.
The only edges within~$N(y^v_{t,i})$ is the edge~$x^v_{t,i}z^v_{t,i}$ and the edge~$x^v_{t+1,j}z^v_{t,i}$ for each~$j\in[\ell]\setminus\{i\}$.
Thus,~$y^v_{t,i}$ is contained in the triangle~$\{y^v_{t,i},z^v_{t,i},x^v_{t,i}\}$ and in the~$\ell-1$ triangles~$\{\{y^v_{t,i},z^v_{t,i}, x^v_{t+1,j}\}\mid j\in[\ell]\setminus\{i\}\}$.
By similar arguments the same holds for vertex~$z^v_{t,i}$.

Thus, each vertex in~$T^v$ is contained in exactly~$\ell$ triangles.\qed
\end{proof}

From Lemma~\ref{lem-l-vertex-triangle-s-club-each-vertex-in-l-triangles} and Observation~\ref{obs-l-vertex-triangle-s-club-triangles-in-vertex-gadgets}, we conclude the following.

\begin{observation}
\label{obs-l-triangle-3-or-4-club-empty-or-complete}
Let~$S$ be a vertex-$\ell$-triangle~$s$-club for~$s=3$ and~$\ell\ge 1$ and or for for~$s\ge 4$ and~$\ell\ge 2$ in~$G'$. 
Then,~$S\cap T^v\neq\emptyset\Leftrightarrow T^v\subseteq S$.
\end{observation}

Now, we prove the correctness of Construction~\ref{const-l-triangle-3-or-4-club}.

\begin{lemma}
\label{lem-l-triangle-3-club-hardness}
For each~$s=3$ and~$\ell\ge 1$, and also for each~$s\ge 4$ and~$\ell\ge 2$ the \textsc{Vertex Triangle~$s$-Club} problem parameterized by~$k$ is W[1]-hard, even if each vertex in the input graph is contained in exactly~$\ell$ triangles.
\end{lemma}

\begin{proof}
We show that~$G$ contains a clique of size at least~$k$ if and only if~$G'$ contains a vertex-$\ell$-triangle~$s$-club of size at least~$3\ell ks^*$.

Let~$K$ be a clique of size at least~$k$ in~$G$.
We argue that~$S\coloneqq \bigcup_{v\in K}T^v$ is a vertex-$\ell$-triangle~$s$-club of size at least~$3\ell ks^*$ in~$G'$.
The size bound follows from the fact that each~$T^v$ consists of exactly~$3\ell s^*$ vertices. 
Furthermore, by Lemma~\ref{lem-l-vertex-triangle-s-club-each-vertex-in-l-triangles} each vertex in~$T^v$  for some~$v\in V(G)$ is contained in exactly~$\ell$ triangles in~$G'[T^v]$.
Hence, it remains to show that~$S$ is an~$s$-club.

To do so, we first prove the following two claims.
To formulate the claims, we need some further notation.
We define~$T^v_0\coloneqq \{p^v_1, \ldots, p^v_{\ell-1}\}\cup\{q^v_1, \ldots, q^v_{\ell-1}\}$ and~$T^v_\ell\coloneqq \{p^v_\ell,q^v_\ell\}$.
Recall that if~$\ell=1$, then~$T^v_0=\emptyset$.
Otherwise, both sets are nonempty.
This claim

\begin{claim}
\label{claim-distances-in-tv}
For~$\ell\ge 2$, we have~$\dist_{G'}(u,a)+\dist_{G'}(u,b)=2s^*$ for each vertex~$a\in T^v_0$, each vertex~$b\in T^v_\ell$, and for each vertex~$u\in T^v\setminus(T^v_0\cup T^v_\ell)$.
\end{claim}

\begin{claimproof}
By~$X^v_1\coloneqq \{x^v_{1,j}\mid j\in[\ell]\}$ we denote the neighbors of~$T^v_0$.
We define the sets~$X^v_2, \ldots, X^v_{s^*}$ and~$Y^v_1, Z^v_1, \ldots, Y^v_{s^*-1},Z^v_{s^*-1}$ analogue.
We first make some observations about the neighborhoods of these vertex sets which then allows us to obtain upper and lower bounds on these paths.
\begin{itemize}
\item $T^v_0\cup\{y^v_{1,i},z^v_{1,i}\}\subseteq N(x^v_{1,i})$ for each~$i\in[\ell]$.
Hence~$N(X^v_1)=T^v_0\cup Y^v_1\cup Z^v_1$.
\item $T^v_\ell\cup\{y^v_{s^*-1,j},z^v_{s^*-1,j}\}\subseteq N(x^v_{s^*,i})$ for each~$i\in[\ell]$ and each~$j\in[\ell]\setminus\{i\}$. 
Hence,~$N(X^v_{s^*})=T^v_\ell\cup Y^v_{s^*-1}\cup Z^v_{s^*-1}$.
\item $\{y^v_{t,i},y^v_{t-1,j},z^v_{t,i},z^v_{t-1,j}\}\subseteq N(x^v_{t,i})$ for each~$t\in[2,s^*-1]$, each~$i\in[\ell]$ and each~$j\in[\ell]\setminus\{i\}$.
Hence,~$N(X^v_t)=Y^v_t\cup Z^v_t\cup Y^v_{t-1}\cup Z^v_{t-1}$ for~$t\in[2,s^*-1]$.
\item $x^v_{t,i}, x^v_{t+1,j},z^v_{t,i}\in N(y^v_{t,i})$ and also~$x^v_{t,i}, x^v_{t+1,j},y^v_{t,i}\in N(z^v_{t,i})$ for each~$t\in[s^*-1]$, each~$i\in[\ell]$ and each~$j\in[\ell]\setminus\{i\}$.
Hence,~$N(Y^v_t)=X^v_t\cup Z^v_t\cup Y^v_{t-1}$ and~$N(Z^v_t)=X^v_t\cup Y^v_t\cup Y^v_{t-1}$ for each~$t\in[s^*-1]$.
\end{itemize}

From the above and a vertex~$u\in X^v_z$ we obtain a path to~$a$ of length~$2z-1$ by using subsequent vertices from~$Y_{z-1},X_{z-1},\ldots, Y_1,X_1$.
A similar observation can be made for a path from~$u$ to~$b$.
The cases~$u\in Y_z\cup Z_z$ is treated similar. 
This implies that~$\dist_{G'}(u,a)+\dist_{G'}(u,b)\le 2s^*$ for each vertex~$a\in T^v_0$, each vertex~$b\in T^v_\ell$, and for each vertex~$u\in T^v\setminus(T^v_0\cup T^v_\ell)$.
Furthermore, note that in each two steps on any path we can increase/decrease the first sub-index by one.
Hence~$\dist_{G'}(u,a)+\dist_{G'}(u,b)\ge 2s^*$.
Thus,~$\dist_{G'}(u,a)+\dist_{G'}(u,b)=2s^*$.
\end{claimproof}

Now, we use Claim~\ref{claim-distances-in-tv} to show that~$T^v$ is a~$(2s^*)$-club.

\begin{claim}
\label{claim-tv-is-club}
If~$s$ is odd, then~$T^v$ is an~$(s-1)$-club and if~$s$ is even, then~$T^v$ is an~$(s-2)$-club for each vertex~$v\in V(G)$.
\end{claim}

\begin{claimproof}
Recall that for~$s\in\{3,4\}$ the gadget~$T^v$ is a non-induced biclique. 
Hence, the statement is true in this case.
In the following, we assume that~$s\ge 5$.
Note that this implies that~$\ell\ge 2$.
We only consider the case that~$s$ is odd. 
The case that~$s$ is even follows analogously.

Consider a pair of vertices~$u,\widetilde{u}$ of~$T^v\setminus(T^v_0\cup T^v_\ell)$.
To bound the distance of~$u$ and~$\widetilde{u}$, we consider one path via a vertex in~$T^v_0$ and one via a vertex in~$T^v_\ell$.
We have $\dist_{G'}(u,\widetilde{u})\le\min(\dist_{G'}(u,a)+\dist_{G'}(a,\widetilde{u}),\dist_{G'}(u,b)+\dist_{G'}(b,\widetilde{u}))$ for each vertex~$a\in T^v_0$, and each vertex~$b\in T^v_\ell$.
From Claim~\ref{claim-distances-in-tv} we know that~$\dist_{G'}(u,a)+\dist_{G'}(u,b)=2s^*$ and that~$\dist_{G'}(\widetilde{u},a)+\dist_{G'}(\widetilde{u},b)=2s^*$.
Hence,~$\dist_{G'}(u,\widetilde{u})\le 2s^*\le s-1$.
Thus,~$T^v$ is indeed an~$(s-1)$-club. 
\end{claimproof}

Now, we show that for two vertices~$v$ and~$w$ in~$K$, a vertex~$u\in T^v$, and a vertex~$\widetilde{u}\in T^w$ we have~$\dist_{G'}(u,\widetilde{u})\le s$.
We consider the three cases:

\textbf{Case I:~$s$ is odd.}
Observe that since~$vw\in E(G)$, each vertex~$u\in (T^v_0\cup T^v_\ell)$ has one neighbor in~$T^w$.
Since~$T^w$ is an~$(s-1)$-club by Claim~\ref{claim-tv-is-club}, we obtain that for each vertex~$\widetilde{u}\in T^w$ we have~$\dist_{G'}(u,\widetilde{u})\le s$ if~$u\in T^v_0\cup T^v_\ell$.
Hence, it remains to consider the case that~$u\in T^v\setminus(T^v_0\cup T^v_\ell)$ and that~$\widetilde{u}\in T^w\setminus(T^w_0\cup T^w_\ell)$.
For this, let~$u^v_1\coloneqq \dist_{G'}(u,p^v_1), \widetilde{u}^w_1\coloneqq \dist_{G'}(q^w_1,\widetilde{u}), u^v_\ell\coloneqq \dist_{G'}(u,p^v_\ell)$, and~$\widetilde{u}^w_\ell\coloneqq \dist_{G'}(q^w_\ell,\widetilde{u})$.
Note that 
\begin{align*}
\dist_{G'}(u,\widetilde{u}) &\le\min(u^v_1+1+\widetilde{u}^w_1,u^v_\ell+1+\widetilde{u}^w_\ell)\\
&=1+\min(u^v_1+\widetilde{u}^w_1,u^v_\ell+\widetilde{u}^w_\ell).
\end{align*}

In this inequality, the '$+1$' is the result of the fact that we have to use an edge to get from gadget~$T^v$ to gadget~$T^w$.
By Claim~\ref{claim-distances-in-tv} we know that~$u^v_1+u^v_\ell= 2s^*$ and that~$\widetilde{u}^w_1+\widetilde{u}^w_\ell= 2s^*$.
Since~$s$ is odd, we obtain that~$\dist_{G'}(u,\widetilde{u})\le 1 + 2s^*=s$.
Thus,~$S$ is a vertex-$\ell$-triangle~$s$-club of size~$3\ell ks^*$.

\textbf{Case II:~$s$ is even and~$\ell\ge 3$.} For each vertex pair~$u\in T^v\setminus T^v_0$ and~$\widetilde{u}\in T^w\setminus T^w_0$ the proof of~$\dist_{G'}(u,\widetilde{u})\le s$ is analogous to the proof in Case~$1$ handling odd values of~$s$.
Hence, it remains to show that each vertex~$u\in T^v_0$ has distance at most~$s$ to each vertex~$\widetilde{u}\in T^w$.
Observe that~$\dist_{G'}(u,p^v_1)\le 2$ since~$T^v_0\subseteq N(x^v_{1,1})$ and~$x_{1,1}^*$ is a neighbor of~$p^v_1$.
Thus,~$\dist_{G'}(u,q^w_1)=3$ since~$vw\in E(G)$.
Furthermore, for each vertex~$\widetilde{u}\in T^w\setminus T^w_\ell$ we have~$\dist_{G'}(\widetilde{u},q^w_1)\le s-3$ by the proof of Claim~\ref{claim-tv-is-club}.
Hence,~$\dist(u,\widetilde{u})\le s$.
Thus, it remains to consider the case that~$u\in T^v_0$ and that~$\widetilde{u}\in T^w_\ell=\{p^w_\ell,q^w_\ell\}$.
Since~$vw\in E(G)$ we conclude that~$\widetilde{u}$ has a neighbor in~$T^v$.
Since~$T^v$ is an~$(s-2)$-club by Claim~\ref{claim-tv-is-club} we conclude that~$\dist_{G'}(u,\widetilde{u})=s-1$.
Hence,~$S$ is a vertex-$\ell$-triangle~$s$-club of size~$3\ell ks^*$.

\textbf{Case III:~$s$ is even and~$\ell=2$.}
Let~$x^w\coloneqq x_{s^*,1}^w$ and~$x^v\coloneqq x_{s^*,1}^v$. 
Furthermore, let~$u^v_1\coloneqq \dist_{G'}(u,p^v_1), \widetilde{u}^w_x\coloneqq \dist_{G'}(x^w,\widetilde{u}), u^v_x\coloneqq \dist_{G'}(u,x^v)$, and~$\widetilde{u}^w_1\coloneqq \dist_{G'}(p^w_1,\widetilde{u})$.
We have
\begin{align*}
\dist_{G'}(u,\widetilde{u})&\le \min(u^v_1+1+\widetilde{u}^w_x,u^v_x+1+\widetilde{u}^w_1)\\
&=1+\min(u^v_1+\widetilde{u}^w_x,u^v_x+\widetilde{u}^w_1).
\end{align*}

Again, the '$+1$' is the result that we have to use an edge to go from~$T^v$ to gadget~$T^w$.
Consider the following claim.

\begin{claim}
\label{claim-s-even-distances}
For each vertex~$u\in T^v$ we have~$u^v_1+u^v_x\le s-1$.
\end{claim}

Claim~\ref{claim-s-even-distances} directly implies that $$\min(u^v_1+\widetilde{u}^w_x,u^v_x+\widetilde{u}^w_1)\le s-1$$ and thus~$\dist_{G'}(u,\widetilde{u})\le s$.
This in turn implies that ~$S$ is a vertex-$\ell$-triangle~$s$-club of size~$3\ell ks^*$.
Thus, it remains to prove Claim~\ref{claim-s-even-distances}.

\begin{claimproof}
Since~$T^v$ is an~$(s-2)$-club by Claim~\ref{claim-tv-is-club}, we conclude that each vertex~$u\in (N[p^v_1]\cup N[x^v_{s^*,1}])\cap T^v$ has distance at most~$s$ to each vertex in~$T^v$.
Hence, it remains to show that each vertex~$u\in T^v\setminus(N[p^v_1]\cup N[x^v_{s^*,1}]))$ has distance at most~$s$ to each vertex~$u'\in T^v$.
If~$s=4$, then~$T^v\subseteq (N[p^v_1]\cup N[x^v_{s^*,1}])$ and hence the claim is proven.
Hence, in the following, we consider the case~$s\ge 6$.

Next, we consider the case that~$u\in T^v_0$.
Since~$T^v_0\subseteq N(x^v_{1,1})$ we conclude that~$u^v_1\le 2$.
By the proof of Claim~\ref{claim-tv-is-club} we conclude that~$u^v_x=s-3$. 
Hence, Claim~\ref{claim-s-even-distances} is true in this case.

For all remaining cases, it is sufficient to prove the claim for each vertex~$u\in T^v\setminus(T^v_0\cup T^v_\ell)$.
First, we consider the case that~$u\coloneqq x^v_{t,i}$ for some~$t\in[s^*]$ and some~$i\in[\ell]$. 
Then we have~$\dist_{G'}(x^v_{t,i},p^v_1)=2t-1$ and~$\dist_{G'}(x^v_{t,i},x^v_{s^*,1})\le 2(s^*-t)+2$ and hence~$\dist_{G'}(x^v_{t,i},p^v_1)+\dist_{G'}(x^v_{t,i},x^v_{s^*,1})\le 2t-1+s^*-2t+2=s^*+1\le s-1$.

Second, we consider the case that~$u\coloneqq y^v_{t,i}$ for some~$t\in[s^*-1]$ and some~$i\in[\ell]$.
Then we have~$\dist_{G'}(y^v_{t,i},p^v_1)=2t$ and~$\dist_{G'}(y^v_{t,i},x^v_{s^*,1})\le 2(s^*-t)+1$ and hence we obtain that~$\dist_{G'}(y^v_{t,i},p^v_1)+\dist_{G'}(y^v_{t,i},x^v_{s^*,1})\le 2t+s^*-2t+1=s^*+1\le s-1$.
The case that~$u\coloneqq z^v_{t,i}$ follows by the same argumentation.
\end{claimproof}

Conversely, suppose that~$G'$ contains a vertex-$\ell$-triangle~$s$-club  of size at least $3\ell ks^*$.
Because of Observation~\ref{obs-l-triangle-3-or-4-club-empty-or-complete} for each vertex-gadget~$T^v$ we either have $T^v\subseteq S$ or~$T^v\cap S=\emptyset$.
Hence,~$S$ contains at least~$k$ vertex gadgets.
We assume towards a contradiction that~$S$ contains two vertex gadgets~$T^v$ and~$T^w$ such that~$vw\notin E(G)$. 
In each case, we will determine a vertex~$u_v\in T^v$ and a vertex~$u_w\in T^w$ such that~$\dist_{G'}(u_v,u_w)\ge s+1$.
This contradiction to the $s$-club property allows us to conclude that the set~$\{v \mid T^v\subseteq S\}$ is a clique of size at least~$k$ in~$G$.

\textbf{Case I:~$s$ is odd.}
We define vertex~$u_v$ as follows. 
The vertex~$u_w$ is defined analogously.
Recall that in this case we have~$s^*=(s-1)/2$. 

\begin{itemize}
\item If~$s\equiv 3\mod 4$, we set~$u_v\coloneqq x^v_{(s+1)/4,1}$. 
\item Otherwise, if~$s\equiv 1\mod 4$, we set~$u_v\coloneqq y^v_{(s-1)/4,1}$.
\end{itemize}

Observe that for each vertex~$u\in T^v_0\cup T^v_\ell$ we have~$\dist_{G'}(u_v,u)=(s-1)/2$.
Furthermore, recall that the vertices in~$T^v_0\cup T^v_\ell$ are the only vertices in~$T^v$ with neighbors in other vertex gadgets.

Similarly, for each vertex~$u'\in T^w_0\cup T^w_\ell$ we have~$\dist_{G'}(u_w,u')=(s-1)/2$.
But since~$vw\notin E(G)$ there are no edges between~$T^v$ and~$T^w$, we obtain for each choice of~$u$ and~$u'$ that~$\dist_{G'}(u_v,u_w)\ge\dist_{G'}(u_v,u)+2+\dist_{G'}(u',u_w)=(s-1)/2+2+(s-1)/2=s+1$, a contradiction.

\textbf{Case II:~$s$ is even and~$\ell\ge 3$.}
We set~$u_v\coloneqq p^v_2$ and~$u_w\coloneqq x^w_{s^*,1}$.
Note that~$u_v\in T^v_0$ since~$\ell\ge 3$.
By the construction we obtain that for each vertex~$u_0\in\{p^v_1,q^v_1\}$ we have~$\dist_{G'}(u_v,u_0)=2$ since~$T^v_0\subseteq N(x^v_{1,1})$.
Furthermore, for each vertex~$u_\ell\in\{p^v_\ell, q^v_\ell\}$ we have~$\dist_{G'}(u_v,u_\ell)=s-2$.
Next, observe that for each vertex~$u'_0\in\{p^w_1, q^w_1\}$ we have~$\dist_{G'}(u_w,u'_0)=s-3$ and for each vertex~$u'_\ell\in\{p^w_\ell, q^w_\ell\}$ we have~$\dist(u_w,u'_\ell)=1$.

Since~$vw\notin E(G)$, there are no edges between the gadgets~$T^v$ and~$T^w$. 
Hence, $\dist(u_v,u_w)\ge\min(\dist(u_v, u_0)+2+\dist_{G'}(u'_0,u_w),\dist_{G'}(u_v, u_\ell)+2+\dist_{G'}(u'_\ell,u_w))$ for each vertex~$u_0\in\{p^v_1, q^v_1\}$, each~$u'_0\in\{p^w_1,q^w_1\}$, each~$u_\ell\in T^v_\ell$, and each~$u'_\ell\in T^w_\ell$.
By the above argumentation we obtain~$\dist_{G'}(u_v,u_w)\ge\min(2+2+s-3,s-2+2+1)=s+1$, a contradiction.

\textbf{Case III:~$s$ is even and~$\ell=2$.}
We define the vertices~$u_v$ and~$u_w$ as follows:

\begin{itemize}
\item If~$s=4$, we set~$u_v\coloneqq x^v_{1,2}$ and~$u_w\coloneqq p^w_2$.
\item If~$s\equiv 0\mod 4$ and~$s\ge 8$, we set~$u_v\coloneqq x^v_{s/4,1}$ and~$u_w\coloneqq y^w_{s/4,1}$.
\item If~$s\equiv 2\mod 8$, we set~$u_v\coloneqq y^v_{(s-2)/4,2}$ and~$u_w\coloneqq x^w_{(s+2)/4,1}$.
\item If~$s\equiv 6\mod 8$, we set~$u_v\coloneqq y^v_{(s-2)/4,1}$ and~$u_w\coloneqq x^w_{(s+2)/4,2}$.
\end{itemize}

From the definition of these vertices we obtain that~$\dist_{G'}(u_v,p^v_1)=s/2-1$, $\dist_{G'}(u_v,x^v_{s^*,1})=s/2$,~$\dist_{G'}(u_w,p^w_1)=s/2$, and that~$\dist_{G'}(u_w,x^w_{s^*,1})=s/2-1$.
Recall that the vertices~$p^u_1$ and~$x^u_{s^*,1}$ are the only vertices in~$T^u$ which have neighbors outside~$T^u$ for each~$u\in V(G)$.
Furthermore, observe that all neighbors of~$p^u_1$ which are not contained in~$T^u$ are the vertices~$x^b_{s^*,1}$ where~$ub\in E(G)$.
Similar, all neighbors of~$x^u_{s^*,1}$ which are not in~$T^u$ are of the form~$p^b_1$ where~$ub\in E*G)$.
We conclude that~$\dist_{G'}(u_v,u_w)\ge \dist(u_v,p^v_1)+3+\dist(x^w_{s^*,1},u_w)$
Here, the '$+3$' results from the fact that at least~$3$ edges to switch the vertex gadgets have to be used: one is not sufficient since~$uw\notin E(G)$ and also two are not sufficient since in two steps one can only reach a vertex~$p^c_1$ for~$c\in V(G)$ from~$p^v_1$ but no vertex~$x^d_{s^*,1}$ for~$d\in V(G)$.
Hence,~$\dist_{G'}(u_v,u_w)\ge s/2-1+3+s/2-1=s+1$, a contradiction.\qed
\end{proof}

\section{Edge Triangle~\texorpdfstring{$s$}{s}-Club}
In this section we settle the parameterized complexity of \textsc{Edge Triangle~$s$-Club} with respect to the solution size~$k$.
Recall that a vertex set~$S$ is an edge-$\ell$-triangle~$s$-club if~$G[S]$ contains a spanning subgraph~$G'=(S,E')$ such that each edge in~$E(G')$ is contained in at least~$\ell$ triangles within~$G'$ and the diameter of~$G'$ is at most~$s$.
First, we show that \textsc{Edge Triangle~$s$-Club} is FPT with respect to~$k$ when~$\ell=1$ irrespective of the value of~$s$ by providing a Turing kernel. 
To show this, it is sufficient to delete edges which are not part of a triangle. 
Afterwards, we prove W[1]-hardness of \textsc{Edge Triangle~$s$-Club} with respect to~$k$ for all fixed~$\ell\ge 2$.

\subsection{Edge Triangle~\texorpdfstring{$s$}{s}-Club with~\texorpdfstring{$\ell=1$}{\ell=1}}

Now, we prove that \textsc{Edge Triangle~$s$-Club} for~$\ell=1$ admits a Turing kernel with respect to~$k$ implying that the problem is FPT. 
To obtain the kernel we need the following reduction rule which removes edges which are in no triangle.

\begin{rrule}
\label{rr-1-edge-triangle-s-club-remove-edges-in-no-triangle}
Let~$(G,k)$ be an instance of \textsc{Edge Triangle~$s$-Club}. 
Delete all edges from~$G$ which are not part of any triangle.
\end{rrule}

It is clear that Reduction Rule~\ref{rr-1-edge-triangle-s-club-remove-edges-in-no-triangle} is correct and can be applied in polynomial time.
The idea that after Reduction Rule~\ref{rr-1-edge-triangle-s-club-remove-edges-in-no-triangle} is applied, we can bound the size of the neighborhood of each vertex.
Next, we prove that after the application of Reduction Rule~\ref{rr-1-edge-triangle-s-club-remove-edges-in-no-triangle} each edge with both endpoints in the closed neighborhood of a vertex in contained in a triangle.

\begin{lemma}
\label{lem-1-edge-triangle-2-club-neighborhood}
Let~$(G,k)$ be an instance of \textsc{Edge Triangle~$s$-Club} with~$\ell=1$ to which Reduction Rule~\ref{rr-1-edge-triangle-s-club-remove-edges-in-no-triangle} is applied.
For each vertex~$v\in V(G)$ each edge in~$G[N[v]]$ is contained in at least one triangle in~$G[N[v]]$.
\end{lemma}

\begin{proof}
First, we consider edges of the form~$uv$ where~$u$ is a neighbor of~$v$.
Since Reduction Rule~\ref{rr-1-edge-triangle-s-club-remove-edges-in-no-triangle} is applied, there exists another vertex~$w\in V(G)$ such that $G[\{u,v,w\}]$ is a triangle. 
Observe that~$w\in N(v)$.

Second, each edge~$uw$ with~$u,w\in N(v)$ is in a triangle with vertex~$v$.\qed
\end{proof}

Next, we show that any instance such that sufficiently many vertices are close to some vertex~$v$ is a yes-instance.

\begin{lemma}
\label{lem-1-edge-triangle-2-club-neighborhood-s-half}
Let~$(G,k)$ be an instance of \textsc{Edge Triangle~$s$-Club} with~$\ell=1$ to which Reduction Rule~\ref{rr-1-edge-triangle-s-club-remove-edges-in-no-triangle} is applied.
Then,~$(G,k)$ is a yes-instance if~$|N_{\lfloor s/2\rfloor}[v]|\ge k$ for some vertex~$v\in V(G)$.
\end{lemma}
\begin{proof}
By Lemma~\ref{lem-1-edge-triangle-2-club-neighborhood} each edge in~$N_{\lfloor s/2\rfloor}[v]$ is contained in at least one triangle since~$N_{\lfloor s/2\rfloor}[v]= \bigcup_{w\in N_{\lfloor s/2\rfloor-1}[v]}N[w]$.
Furthermore, each vertex in~$N_{\lfloor s/2\rfloor}[v]$ has distance at most~$\lfloor s/2\rfloor$ to vertex~$v$.
Hence,~$N_{\lfloor s/2\rfloor}[v]$ is an~$s$-club and by definition~$|N_{\lfloor s/2\rfloor}[v]|\ge k$.\qed
\end{proof}

Lemma~\ref{lem-1-edge-triangle-2-club-neighborhood-s-half} implies a Turing kernel for~$k$ which implies that the problem is fixed-parameter tractable. 

\begin{theorem}
\label{thm-1-edge-triangle-s-club-fpt}
\textsc{Edge Triangle~$s$-Club} for~$\ell=1$ admits a~$k^2$-vertex Turing kernel if~$s$ is even and  a~$k^3$-vertex Turing kernel if~$s$ is odd and~$s\ge 3$.
\end{theorem}
\begin{proof}
First, we apply Reduction Rule~\ref{rr-1-edge-triangle-s-club-remove-edges-in-no-triangle} exhaustively.
Because of Lemma~\ref{lem-1-edge-triangle-2-club-neighborhood-s-half} we conclude that~$(G,k)$ is a trivial yes-instance, if~$|N_{\lfloor s/2\rfloor}[v]|\ge k$ for some~$v\in V(G)$.
Hence, in the following we can assume that~$k> |N_{\lfloor s/2\rfloor}[v]|$ for each vertex~$v\in V(G)$.

First, we consider the case that~$s$ is even. 
Then~$\lfloor s/2\rfloor=s/2$ and we obtain that~$N_s[v]\subseteq N_{s/2}[N_{s/2}[v]]$ for each~$v\in V(G)$.
Thus,~$|N_s[v]|\le k^2$.

Second, we consider the case that~$s$ is odd.
Observe that we have $N_s[v]\subseteq N_{\lfloor s/2\rfloor}[N_{\lfloor s/2\rfloor}[N_{\lfloor s/2\rfloor}[v]]]$ for each~$v\in V(G)$.
Thus,~$|N_s[v]|\le k^3$.\qed
\end{proof}

\subsection{Edge Triangle~\texorpdfstring{$s$}{s}-Club for~\texorpdfstring{$\ell\ge 2$}{\ell\ge 2}}

Now we show W[1]-hardness for the remaining cases.

\begin{theorem}
\label{thm-ell-edge-triangle-s-club-w-hard}
\textsc{Edge Triangle~$s$-Club} is W[1]-hard for parameter~$k$ if~$\ell\ge 2$.
\end{theorem}

Next, we describe the construction of the reduction to prove Theorem~\ref{thm-ell-edge-triangle-s-club-w-hard}.
We reduce from \textsc{Clique}.
The idea is to construct one vertex gadget for each vertex of the \textsc{Clique} instance and to add edges between two different vertex gadgets if and only of the two corresponding vertices are adjacent in such a way that all these edges are in exactly~$\ell$ triangles.
For an illustration of this construction see Fig.~\ref{fig-visulation-edge-triangle}.

\begin{const}
\label{const-ell-triangle-s-club-w-hardness}
Let~$(G,k)$ be an instance of \textsc{Clique} with~$k\ge 3$.
We construct an equivalent instance~$(G',k')$ of \textsc{Edge-Triangle~$s$-Club} for some fixed~$\ell\ge 2$ as follows.
Let~$\ell^*\coloneqq \lceil\ell/2\rceil$ and let~$x\coloneqq 6\cdot \ell^* (s-1)+\lfloor\ell/2\rfloor$.
For each vertex~$v\in V(G)$, we construct the following vertex gadget~$T^v$. 
For better readability, all sub-indices of the vertices in~$T^v$ are considered modulo~$x$.
Our construction distinguishes even and odd values of~$\ell$.
First, we describe the part of the construction which both cases have in common.

\begin{enumerate}
\item \label{item-edge-variant-1}
We add vertex sets~$A_v\coloneqq \{a^v_i\mid i\in[0,x]\}$ and~$B_v\coloneqq \{b^v_i\mid i\in[0,x]\}$ to~$G'$.

\item \label{item-edge-variant-2}
We add the edges~$a^v_ia^v_{i+j}$, and~$b^v_ib^v_{i+j}$ for each~$i\in[0,x]$ and each~$j\in[-3\ell^*,3\ell^*]\setminus\{0\}$ to~$G'$.

\item \label{item-edge-variant-3}
We add the edge~$a^v_ib^v_{i+j}$ for each~$i\in[0,x]$ and each~$j\in[-3\ell^*,3\ell^*]$ to~$G'$.
\end{enumerate}

In other words, an edge~$a^v_i b^v_j$ is added if the indices differ by at most~$3\ell^*$.
For even~$\ell$, this completes the construction of~$T^v$.
For odd~$\ell$, we extend~$T^v$ as follows:

\begin{enumerate}[label=0-\arabic*.,leftmargin=*]
\item \label{item-edge-variant-4}
We add the vertex set~$C_v\coloneqq \{c^v_{i}\mid i\in[0,x] \text{ and } i\equiv 0\mod\ell^*\}$ to~$G'$.
Note that~$C_v$ consists of exactly~$6s-5$ vertices.

\item \label{item-edge-variant-5}
We add the edges~$c^v_ia^v_{i+j}$ and~$c^v_ib^v_{i+j}$ for each~$i\in[0,x]$ such that~$i\equiv 0\mod\ell^*$ and each~$j\in[-3\ell^*,3\ell^*]$ to~$G'$.

\item \label{item-edge-variant-6}
Also, we add the edge~$c^v_ic^v_{i+j}$ to~$G'$ for each~$i\in[0,x]$ such that~$i\equiv 0\mod\ell^*$ and each~$j\in[-3\ell^*,3\ell^*]\setminus\{0\}$ to~$G'$ if the corresponding vertex~$c^v_{i+j}$ exists.
\end{enumerate}

In other words, an edge between~$c^v_i$ and~$a^v_j$,~$b^v_j$, or~$c^v_j$ is added if the indices differ by at most~$3\ell^*$.
Now, for each edge~$uv\in E(G)$, we add the following to~$G'$:
\begin{enumerate}[label=0-\arabic*.,leftmargin=*, start=4]
\item \label{item-edge-variant-7} 
We add the edges~$a^v_ib^u_{i+j}$ and~$a^u_ib^v_{i+j}$ for each~$i\in[0,x]$ and~$j\in[0,\lfloor \ell/2\rfloor]$.

\item \label{item-edge-variant-8} 
If~$\ell$ is odd, we also add the edges~$c^v_ib^u_{i+j}$ and~$c^u_ib^v_{i+j}$ for each~$i\in[0,x]$ such that~$i\equiv 0\mod\ell^*$ and each~$j\in[0,\lfloor \ell/2\rfloor]$ to~$G'$.

Observe that each vertex~$b^u_{i+j}$ is adjacent to \emph{exactly} one vertex in~$C_v$.
\end{enumerate}

In other words, an edge between~$a^v_i$ or~$c^v_i$ and~$b^v_j$ is added if~$j$ exceeds~$i$ by at most~$\lfloor\ell/2\rfloor$.
Finally, if~$\ell$ is even, we set~$k'\coloneqq 2(x+1)k=(\ell(6s-5)+2)\cdot k$, and if~$\ell$ is odd, we set~$k'\coloneqq (2(x+1)+6s-5)k=(\ell+2)(6s-5)\cdot k$.
\end{const}

\begin{figure}[t]
\begin{tikzpicture}[scale=0.85]
\node[label=above:{$a^v_0$}](av0) at (0, 0) [shape = circle, draw, fill=black, scale=0.07ex]{};
\node[label=above:{$a^v_1$}](av1) at (1, 0) [shape = circle, draw, fill=black, scale=0.07ex]{};
\node[label=above:{$a^v_2$}](av2) at (2, 0) [shape = circle, draw, fill=black, scale=0.07ex]{};
\node[label=above:{$a^v_3$}](av3) at (3, 0) [shape = circle, draw, fill=black, scale=0.07ex]{};
\node[label=above:{$a^v_4$}](av4) at (4, 0) [shape = circle, draw, fill=black, scale=0.07ex]{};
\node[label=above:{$a^v_5$}](av5) at (5, 0) [shape = circle, draw, fill=black, scale=0.07ex]{};
\node[label=above:{$a^v_6$}](av6) at (6, 0) [shape = circle, draw, fill=black, scale=0.07ex]{};
\node[label=above:{$a^v_7$}](av7) at (7, 0) [shape = circle, draw, fill=black, scale=0.07ex]{};
\node[label=above:{$a^v_8$}](av8) at (8, 0) [shape = circle, draw, fill=black, scale=0.07ex]{};
\node[label=above:{$a^v_9$}](av9) at (9, 0) [shape = circle, draw, fill=black, scale=0.07ex]{};
\node[label=above:{$a^v_{10}$}](av10) at (10, 0) [shape = circle, draw, fill=black, scale=0.07ex]{};
\node[label=above:{$a^v_{11}$}](av11) at (11, 0) [shape = circle, draw, fill=black, scale=0.07ex]{};
\node[label=above:{$a^v_{12}$}](av12) at (12, 0) [shape = circle, draw, fill=black, scale=0.07ex]{};
\node[label=above:{$a^v_{13}$}](av13) at (13, 0) [shape = circle, draw, fill=black, scale=0.07ex]{};
\draw[draw=black] (-0.3,-0.2) rectangle ++(13.6,0.8);
\node[label=right:{$A_v$}](av) at (13.1, 0.2) {};

\draw[draw=blue] (1,0) circle (4pt);
\draw[draw=blue] (2,0) circle (4pt);
\draw[draw=blue] (3,0) circle (4pt);
\draw[draw=blue] (11,0) circle (4pt);
\draw[draw=blue] (12,0) circle (4pt);
\draw[draw=blue] (13,0) circle (4pt);
\draw[draw=red] (4,0) circle (4pt);
\draw[draw=red] (5,0) circle (4pt);
\draw[draw=red] (6,0) circle (4pt);
\draw[draw=red] (8,0) circle (4pt);
\draw[draw=red] (9,0) circle (4pt);
\draw[draw=red] (10,0) circle (4pt);
\draw[draw=black] (7,0) circle (4pt);

\node(bw0) at (0, -1) [shape = circle, draw, fill=black, scale=0.07ex]{};
\node(bw1) at (1, -1) [shape = circle, draw, fill=black, scale=0.07ex]{};
\node(bw2) at (2, -1) [shape = circle, draw, fill=black, scale=0.07ex]{};
\node(bw3) at (3, -1) [shape = circle, draw, fill=black, scale=0.07ex]{};
\node(bw4) at (4, -1) [shape = circle, draw, fill=black, scale=0.07ex]{};
\node(bw5) at (5, -1) [shape = circle, draw, fill=black, scale=0.07ex]{};
\node(bw6) at (6, -1) [shape = circle, draw, fill=black, scale=0.07ex]{};
\node(bw7) at (7, -1) [shape = circle, draw, fill=black, scale=0.07ex]{};
\node(bw8) at (8, -1) [shape = circle, draw, fill=black, scale=0.07ex]{};
\node(bw9) at (9, -1) [shape = circle, draw, fill=black, scale=0.07ex]{};
\node(bw10) at (10, -1) [shape = circle, draw, fill=black, scale=0.07ex]{};
\node(bw11) at (11, -1) [shape = circle, draw, fill=black, scale=0.07ex]{};
\node(bw12) at (12, -1) [shape = circle, draw, fill=black, scale=0.07ex]{};
\node(bw13) at (13, -1) [shape = circle, draw, fill=black, scale=0.07ex]{};
\draw[draw=black] (-0.3,-1.15) rectangle ++(13.6,0.3);
\node[label=right:{$B_w$}](av) at (13.1, -1) {};

\draw[draw=blue] (0,-1) circle (4pt);
\draw[draw=blue] (1,-1) circle (4pt);
\draw[draw=red] (2,-1) circle (4pt);
\draw[draw=red] (3,-1) circle (4pt);
\draw[draw=red] (4,-1) circle (4pt);
\draw[draw=red] (11,-1) circle (4pt);
\draw[draw=red] (12,-1) circle (4pt);
\draw[draw=red] (13,-1) circle (4pt);
\draw[draw=black] (5,-1) circle (4pt);
\draw[draw=black] (6,-1) circle (4pt);
\draw[draw=black] (7,-1) circle (4pt);
\draw[draw=black] (8,-1) circle (4pt);
\draw[draw=black] (9,-1) circle (4pt);
\draw[draw=black] (10,-1) circle (4pt);

\node(aw0) at (0, -2) [shape = circle, draw, fill=black, scale=0.07ex]{};
\node(aw1) at (1, -2) [shape = circle, draw, fill=black, scale=0.07ex]{};
\node(aw2) at (2, -2) [shape = circle, draw, fill=black, scale=0.07ex]{};
\node(aw3) at (3, -2) [shape = circle, draw, fill=black, scale=0.07ex]{};
\node(aw4) at (4, -2) [shape = circle, draw, fill=black, scale=0.07ex]{};
\node(aw5) at (5, -2) [shape = circle, draw, fill=black, scale=0.07ex]{};
\node(aw6) at (6, -2) [shape = circle, draw, fill=black, scale=0.07ex]{};
\node(aw7) at (7, -2) [shape = circle, draw, fill=black, scale=0.07ex]{};
\node(aw8) at (8, -2) [shape = circle, draw, fill=black, scale=0.07ex]{};
\node(aw9) at (9, -2) [shape = circle, draw, fill=black, scale=0.07ex]{};
\node(aw10) at (10, -2) [shape = circle, draw, fill=black, scale=0.07ex]{};
\node(aw11) at (11, -2) [shape = circle, draw, fill=black, scale=0.07ex]{};
\node(aw12) at (12, -2) [shape = circle, draw, fill=black, scale=0.07ex]{};
\node(aw13) at (13, -2) [shape = circle, draw, fill=black, scale=0.07ex]{};
\draw[draw=black] (-0.3,-2.15) rectangle ++(13.6,0.3);
\node[label=right:{$A_w$}](av) at (13.1, -2) {};

\draw[draw=red] (0,-2) circle (4pt);
\draw[draw=red] (1,-2) circle (4pt);
\draw[draw=red] (2,-2) circle (4pt);
\draw[draw=red] (3,-2) circle (4pt);
\draw[draw=red] (4,-2) circle (4pt);
\draw[draw=red] (11,-2) circle (4pt);
\draw[draw=red] (12,-2) circle (4pt);
\draw[draw=red] (13,-2) circle (4pt);
\draw[draw=black] (5,-2) circle (4pt);
\draw[draw=black] (6,-2) circle (4pt);
\draw[draw=black] (7,-2) circle (4pt);
\draw[draw=black] (8,-2) circle (4pt);
\draw[draw=black] (9,-2) circle (4pt);
\draw[draw=black] (10,-2) circle (4pt);

\node[label=below:{$b^u_0$}](bu0) at (0, -3) [shape = circle, draw, fill=black, scale=0.07ex]{};
\node[label=below:{$b^u_1$}](bu1) at (1, -3) [shape = circle, draw, fill=black, scale=0.07ex]{};
\node[label=below:{$b^u_2$}](bu2) at (2, -3) [shape = circle, draw, fill=black, scale=0.07ex]{};
\node[label=below:{$b^u_3$}](bu3) at (3, -3) [shape = circle, draw, fill=black, scale=0.07ex]{};
\node[label=below:{$b^u_4$}](bu4) at (4, -3) [shape = circle, draw, fill=black, scale=0.07ex]{};
\node[label=below:{$b^u_5$}](bu5) at (5, -3) [shape = circle, draw, fill=black, scale=0.07ex]{};
\node[label=below:{$b^u_6$}](bu6) at (6, -3) [shape = circle, draw, fill=black, scale=0.07ex]{};
\node[label=below:{$b^u_7$}](bu7) at (7, -3) [shape = circle, draw, fill=black, scale=0.07ex]{};
\node[label=below:{$b^u_8$}](bu8) at (8, -3) [shape = circle, draw, fill=black, scale=0.07ex]{};
\node[label=below:{$b^u_9$}](bu9) at (9, -3) [shape = circle, draw, fill=black, scale=0.07ex]{};
\node[label=below:{$b^u_{10}$}](bu10) at (10, -3) [shape = circle, draw, fill=black, scale=0.07ex]{};
\node[label=below:{$b^u_{11}$}](bu11) at (11, -3) [shape = circle, draw, fill=black, scale=0.07ex]{};
\node[label=below:{$b^u_{12}$}](bu12) at (12, -3) [shape = circle, draw, fill=black, scale=0.07ex]{};
\node[label=below:{$b^u_{13}$}](bu13) at (13, -3) [shape = circle, draw, fill=black, scale=0.07ex]{};
\draw[draw=black] (-0.3,-3.6) rectangle ++(13.6,0.8);
\node[label=right:{$B_u$}](av) at (13.1, -3.2) {};

\draw[draw=black] (0,-3) circle (4pt);
\draw[draw=black] (1,-3) circle (4pt);
\draw[draw=black] (2,-3) circle (4pt);
\draw[draw=black] (3,-3) circle (4pt);
\draw[draw=black] (4,-3) circle (4pt);
\draw[draw=black] (5,-3) circle (4pt);
\draw[draw=black] (11,-3) circle (4pt);
\draw[draw=black] (12,-3) circle (4pt);
\draw[draw=black] (13,-3) circle (4pt);

\path [-,line width=0.2mm] (av0) edge (bw0);
\path [-,line width=0.2mm] (av0) edge (bw1);
\path [-,line width=0.2mm] (av1) edge (bw1);
\path [-,line width=0.2mm] (av1) edge (bw2);
\path [-,line width=0.2mm] (av2) edge (bw2);
\path [-,line width=0.2mm] (av2) edge (bw3);
\path [-,line width=0.2mm] (av3) edge (bw3);
\path [-,line width=0.2mm] (av3) edge (bw4);
\path [-,line width=0.2mm] (av4) edge (bw4);
\path [-,line width=0.2mm] (av4) edge (bw5);
\path [-,line width=0.2mm] (av5) edge (bw5);
\path [-,line width=0.2mm] (av5) edge (bw6);
\path [-,line width=0.2mm] (av6) edge (bw6);
\path [-,line width=0.2mm] (av6) edge (bw7);
\path [-,line width=0.2mm] (av7) edge (bw7);
\path [-,line width=0.2mm] (av7) edge (bw8);
\path [-,line width=0.2mm] (av8) edge (bw8);
\path [-,line width=0.2mm] (av8) edge (bw9);
\path [-,line width=0.2mm] (av9) edge (bw9);
\path [-,line width=0.2mm] (av9) edge (bw10);
\path [-,line width=0.2mm] (av10) edge (bw10);
\path [-,line width=0.2mm] (av10) edge (bw11);
\path [-,line width=0.2mm] (av11) edge (bw11);
\path [-,line width=0.2mm] (av11) edge (bw12);
\path [-,line width=0.2mm] (av12) edge (bw12);
\path [-,line width=0.2mm] (av12) edge (bw13);
\path [-,line width=0.2mm] (av13) edge (bw13);
\path [-,line width=0.2mm] (av13) edge (bw0);

\path [-,line width=0.2mm] (aw0) edge (bu0);
\path [-,line width=0.2mm] (aw0) edge (bu1);
\path [-,line width=0.2mm] (aw1) edge (bu1);
\path [-,line width=0.2mm] (aw1) edge (bu2);
\path [-,line width=0.2mm] (aw2) edge (bu2);
\path [-,line width=0.2mm] (aw2) edge (bu3);
\path [-,line width=0.2mm] (aw3) edge (bu3);
\path [-,line width=0.2mm] (aw3) edge (bu4);
\path [-,line width=0.2mm] (aw4) edge (bu4);
\path [-,line width=0.2mm] (aw4) edge (bu5);
\path [-,line width=0.2mm] (aw5) edge (bu5);
\path [-,line width=0.2mm] (aw5) edge (bu6);
\path [-,line width=0.2mm] (aw6) edge (bu6);
\path [-,line width=0.2mm] (aw6) edge (bu7);
\path [-,line width=0.2mm] (aw7) edge (bu7);
\path [-,line width=0.2mm] (aw7) edge (bu8);
\path [-,line width=0.2mm] (aw8) edge (bu8);
\path [-,line width=0.2mm] (aw8) edge (bu9);
\path [-,line width=0.2mm] (aw9) edge (bu9);
\path [-,line width=0.2mm] (aw9) edge (bu10);
\path [-,line width=0.2mm] (aw10) edge (bu10);
\path [-,line width=0.2mm] (aw10) edge (bu11);
\path [-,line width=0.2mm] (aw11) edge (bu11);
\path [-,line width=0.2mm] (aw11) edge (bu12);
\path [-,line width=0.2mm] (aw12) edge (bu12);
\path [-,line width=0.2mm] (aw12) edge (bu13);
\path [-,line width=0.2mm] (aw13) edge (bu13);
\path [-,line width=0.2mm] (aw13) edge (bu0);

\path [-,line width=0.2mm] (aw3) edge (bw0);
\path [-,line width=0.2mm] (aw3) edge (bw1);
\path [-,line width=0.2mm] (aw3) edge (bw2);
\path [-,line width=0.2mm] (aw3) edge (bw3);
\path [-,line width=0.2mm] (aw3) edge (bw4);
\path [-,line width=0.2mm] (aw3) edge (bw5);
\path [-,line width=0.2mm] (aw3) edge (bw6);

\path [-,line width=0.2mm] (aw4) edge (bw1);
\path [-,line width=0.2mm] (aw4) edge (bw2);
\path [-,line width=0.2mm] (aw4) edge (bw3);
\path [-,line width=0.2mm] (aw4) edge (bw4);
\path [-,line width=0.2mm] (aw4) edge (bw5);
\path [-,line width=0.2mm] (aw4) edge (bw6);
\path [-,line width=0.2mm] (aw4) edge (bw7);

\path [-,line width=0.2mm] (aw9) edge (bw6);
\path [-,line width=0.2mm] (aw9) edge (bw7);
\path [-,line width=0.2mm] (aw9) edge (bw8);
\path [-,line width=0.2mm] (aw9) edge (bw9);
\path [-,line width=0.2mm] (aw9) edge (bw10);
\path [-,line width=0.2mm] (aw9) edge (bw11);
\path [-,line width=0.2mm] (aw9) edge (bw12);

\path [-,line width=0.2mm] (aw10) edge (bw7);
\path [-,line width=0.2mm] (aw10) edge (bw8);
\path [-,line width=0.2mm] (aw10) edge (bw9);
\path [-,line width=0.2mm] (aw10) edge (bw10);
\path [-,line width=0.2mm] (aw10) edge (bw11);
\path [-,line width=0.2mm] (aw10) edge (bw12);
\path [-,line width=0.2mm] (aw10) edge (bw13);

\end{tikzpicture}
\caption{Construction for Theorem~\ref{thm-ell-edge-triangle-s-club-w-hard} when~$s=3$ and~$\ell=2$ and~$G$ is a~$P_3$ on~$\{u,v,w\}$ with~$uv\notin E(G)$. 
Only the gadgets~$A_v, B_w, A_w$, and~$B_u$ are shown.
For simplicity, no edges within~$A_v$, $B_w$, $A_w$, and~$B_u$ are drawn and edges between~$B_w$ and~$A_w$ are only drawn if one endpoint is~$a^w_3$, $a^w_4$, $a^w_9$, or~$a^w_{10}$.\newline
Blue encircled vertices are neighbors of~$a^v_0$, red encircled vertices have distance~$2$ to~$a^v_0$, and black encircled vertices have distance~$3$ to~$a^v_0$.
Thus,~$a^v_0$ and~$b^u_7=b^u_{0+1+3\cdot 1\cdot 2}=b^u_{0+\lfloor\ell/2\rfloor+3\ell^*(s-1)}$ have distance at least~$4$.}
\label{fig-visulation-edge-triangle}
\end{figure}
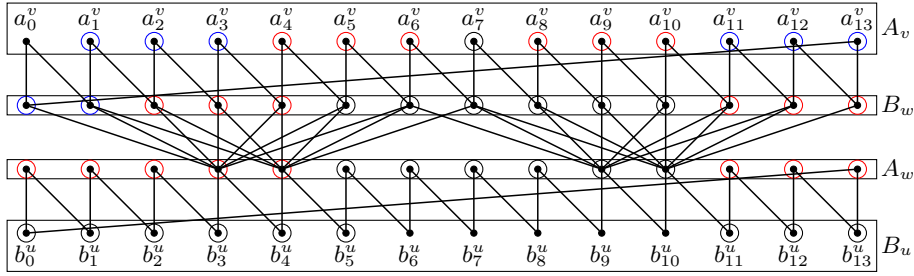

Construction~\ref{const-ell-triangle-s-club-w-hardness} has two key mechanisms:
First, if~$uv\notin E(G)$ then for each vertex~$a\in A_v$ there is at least one vertex~$b\in B_u$ such that~$\dist(a,b)>s$. 
Second, each edge with one endpoint in~$A_v$ and one endpoint in~$B_u$ is contained in \emph{exactly}~$\ell$ triangles.
Furthermore, if~$\ell$ is odd, then this also holds for each edge with one endpoint in~$C_v$ and one in~$B_u$.
Consider an edge-$\ell$-triangle~$s$-club~$S$ and let~$\widetilde{G}=(S,\widetilde{E})$ be a spanning subgraph of~$G[S]$ with the maximal number of edges, such that each edge of~$\widetilde{E}$ is contained in at least~$\ell$ triangles in~$\widetilde{G}$ and the diameter of~$\widetilde{G}$ is~$s$.
As we will show, the two mechanics ensure that an edge with one endpoint in~$A_v$ (or~$C_v$) and the other endpoint in~$B_u$ is contained in~$\widetilde{E}$ if and only if~$S$ contains all vertices of~$A_v$ (and~$C_v$) and~$B_u$.
We call this the \emph{enforcement property}. 
Next, we formalize this property.
To this end, we introduce the following notation.
By~$E_{uv}$ we denote the set of all edges with one endpoint in~$A_v$ (or~$C_v$ if~$\ell$ is odd) and the other endpoint in~$B_u$.

\begin{lemma}
\label{lem-ell-triangle-s-club-av-also-bu}
Let~$S$ be an edge-$\ell$-triangle~$s$-club~$G'$ constructed in Construction~\ref{const-ell-triangle-s-club-w-hardness}.
More precisely, let~$\widetilde{G}=(S,\widetilde{E})$ be a maximal subgraph of~$G[S]$ such that each edge in~$E(\widetilde{G})$ is contained in at least~$\ell$ triangles within~$\widetilde{G}$ and the diameter of~$\widetilde{G}$ is at most~$s$. 
Let~$e\in E_{uv}$.
Then~$e\in E(\widetilde{G})$ if and only if~$A_v,B_u\subseteq S$ (and~$C_v\subseteq S$, if~$\ell$ is odd).
\end{lemma}
\begin{proof}
Before we show the two implications, we prove the following cascading property of edge-$\ell$-triangle~$s$-clubs which contain at least one edge of~$E_{uv}$.

\begin{claim}
\label{claim-edge-variant-cascading-av-bu}
If~$a^v_ib^u_j\in E(\widetilde{G})$ or~$c^v_ib^u_j\in E(\widetilde{G})$, then,$E_{uv}\subseteq E(\widetilde{G})$.
\end{claim}

\begin{claimproof}
First, we consider even values of~$\ell$.
Note that~$\ell^*=\lfloor\ell/2\rfloor=\ell/2=\lceil\ell/2\rceil$.
By construction we have
$$N[a^v_i]= \{a^v_{i+i'}, b^v_{i+i'}\mid i'\in[-3\ell/2,3\ell/2]\}\cup\{b^w_{i+i'}\mid i'\in[0,\ell/2] \text{ and }vw\in E(G)\},$$ 
and similar
$$N[b^u_j]=\{a^u_{j+i'},b^u_{j+i'}\mid i'\in[-3\ell/2,3\ell/2]\}\cup\{a^w_{j-i'}\mid i'\in[0,\ell/2] \text{ and }uw\in E(G)\}.$$

Since~$a^v_ib^u_j\in E(G')$ we obtain by Part~\ref{item-edge-variant-7} of Construction~\ref{const-ell-triangle-s-club-w-hardness} that~$j=i+z$ for some~$z\in[0,\ell/2]$.
Let~$y\coloneqq\ell/2-z$ and observe that~$N(a^v_i)\cap N(b^u_j)=\{a^v_{i+i'}\mid i'\in[-y,z]\setminus\{0\}\}\cup\{b^u_{j+i'} \mid i'\in[-z,y]\setminus{\{0\}}\}$.
Thus, the edge~$a^v_ib^u_j$ is contained in exactly~$\ell$ triangles whose vertex sets are all contained in~$G'[A_v\cup B_u]$.
Since~$a^v_ib^u_j\in E(\widetilde{G})$, we thus conclude that all vertices and edges which form these $\ell$~triangles are contained in~$\widetilde{G}$.
In particular, we obtain that~$a^v_{i+1}b^u_{j+1}\in E(\widetilde{G})$ and thus also that~$a^v_{i+1},b^u_{j+1}\in V(\widetilde{G})$.
Now, since~$a^v_{i+1}b^u_{j+1}\in E(\widetilde{G})$ we can repeat the above argumentation for the edge~$a^v_{i+1}b^u_{j+1}$ and inductively for the edge~$a^v_{i+q}b^v_{i+q}$ for all~$q\in [x]$.
We then have verified that~$A_v\cup B_u\subseteq V(\widetilde{G})$ and that each edge in~$E_{uv}$ is contained in~$E(\widetilde{G})$.
Recall that~$A_v$ and~$B_u$ have size~$x+1$.

Second, we consider odd values of~$\ell$, that is~$\ell=2t+1$ for some integer~$t$.
Note that~$\lfloor\ell/2\rfloor=t$ and that~$\ell^*=\lceil\ell/2\rceil=t+1$.
Furthermore, observe that vertex~$b^u_j$ has exactly one neighbor~$c^v_{j+i'}$ in~$C_v$ for some~$i'\in[0,t]$ such that~$j+i'\mod (t+1)=0$.

Hence, by construction we have
\begin{itemize}
\item $N[a^v_i]=N[c^v_i]=\{a^v_{i+i'},b^v_{i+i'}\mid i'\in[-3(t+1),3(t+1)]\}\cup\{c^v_{i+i'}\mid i'\in[-3(t+1),3(t+1)] \text{ and } (i+i')\mod (t+1)=0\}\cup\{b^w_{i+i'}\mid i'\in[0,t] \text{ and }vw\in E(G)\}$,
\item $N[b^u_j]=\{a^u_{j+i'},b^u_{j+i'}\mid i'\in[-3(t+1),3(t+1)]\}\cup\{c^u_{j+i'}\mid i'\in[-3(t+1),3(t+1)] \text{ and } (j+i')\mod (t+1)=0\}\cup\{a^w_{j-i'}\mid i'\in[0,t] \text{ and }uw\in E(G)\}\cup \{c^w_{j-i'}\mid i'\in [0,t] \text{ and } (j-i')\mod (t+1)=0 \text{ and }uw\in E(G)\}$, and
\end{itemize}

Since~$a^v_ib^u_j\in E(G')$ we obtain by Part~\ref{item-edge-variant-7} of Construction~\ref{const-ell-triangle-s-club-w-hardness} that~$j=i+z$ for some~$z\in[0,t]$.
Now, let~$y\coloneqq t-z$ and let~$c^v_{i'}$ be the unique neighbor of~$b^u_j$ in~$C_v$.
We conclude that~$N(a^v_i)\cap N(b^u_j)=\{a^v_{i+j'}\mid j'\in[-y,z]\setminus\{0\}\}\cup\{b^u_{j+j'} \mid j'\in[-z,y]\setminus{0}\}\cup\{c^v_{i'}\}$.

Hence, both vertices have exactly~$t+t+1=\ell$ common neighbors and thus~$N(a^v_i)\cap N(b^u_j)\subseteq S$.
By similar arguments a similar statement can be shown for the edge~$c^v_ib^u_j$.
Thus, the edges~$a^v_ib^u_j$ and~$c^v_ib^u_j$ are contained in exactly~$\ell$ triangles whose vertex sets are contained in~$G[A_v\cup C_v\cup B_u]$.

Since~$a^v_ib^u_j\in E(\widetilde{G})$ we thus conclude that all vertices and edges which form these~$\ell$ triangles are contained in~$\widetilde{G}$.
In particular, we obtain that~$a^v_{i+1}b^u_{j+1}\in E(\widetilde{G})$ and also that~$a^v_{i+1},b^u_{j+1}\in V(\widetilde{G})$.
Now, since~$a^v_{i+1}b^u_{j+1}\in E(\widetilde{G})$ we can repeat the above argumentation for the edge~$a^v_{i+1}b^u_{j+1}$ and inductively for the edge~$a^v_{i+q}b^v_{i+q}$ for all~$q\in [x]$.
We then have verified that~$A_v\cup B_u\cup C_v\subseteq V(\widetilde{G})$ and that each edge in~$E_{uv}$ is contained in~$E(\widetilde{G})$.
\end{claimproof}

Now, we are ready to prove the two implications.
The implication that if~$e\in E(\widetilde{G})$ then~$A_v,B_u\subseteq S$ (and~$C_v\subseteq S$, if~$\ell$ is odd) directly follows from Claim~\ref{claim-edge-variant-cascading-av-bu}.
It remains to show the other implication.
From Claim~\ref{claim-edge-variant-cascading-av-bu} we conclude that either~$E_{uv}\subseteq E(\widetilde{G})$ or~$E_{uv}\cap E(\widetilde{G})=\emptyset$.
If~$E_{uv}\subseteq E(\widetilde{G})$, then we are done, so assume towards a contradiction that~$E_{uv}\cap E(\widetilde{G})=\emptyset$.
Recall that~$A_v\cup B_u\subseteq S$ (and also~$C_v\subseteq S$ is~$\ell$ is odd).
Furthermore, recall that~$\widetilde{G}$ is maximal, that is, there exists no spanning subgraph of~$G[S]$ which has more edges than~$\widetilde{G}$.
Hence, adding all edges in~$E_{uv}$ to~$\widetilde{G}$ is still an edge-$\ell$-triangle~$s$-club, a contradiction to the maximality of~$\widetilde{G}$.\qed
\end{proof}

Now, we prove the correctness of the reduction for Theorem~\ref{thm-ell-edge-triangle-s-club-w-hard}.

\begin{proof}[of Theorem~\ref{thm-ell-edge-triangle-s-club-w-hard}]
We show that~$G$ contains a clique of size at least~$k$ if and only if~$G'$ contains an edge-$\ell$-triangle~$s$-club of size at least~$k'$.

Let~$K$ be a clique of size~$k$ in~$G$.
Recall that~$T^v$ is the gadget of vertex~$v\in V(G)$.
We verify that~$S\coloneqq\{u\in V(T^v)\mid v\in K\}$ is an edge-$\ell$-triangle~$s$-club of size at least~$k'$.
More precisely, we show that~$\widetilde{G}\coloneqq G[S]$ fulfills all properties of being an edge-$\ell$-triangle~$s$-club.
Since~$|T^v|=2\ell^*(6s-5)+2$ if~$\ell$ is even and~$|T^v|=(2\ell^*+1)(6s-5)$ if~$\ell$ is odd for each~$v\in K$ and since~$|K|\ge k$, we have~$|S|=|V(\widetilde{G})|\ge k'$.
It remains to show that~$\widetilde{G}$ is an edge-$\ell$-triangle~$s$-club.

\textbf{Next, we show that~$\widetilde{G}$ is an~$s$-club.}

First, we show that~$T^v$ is an~$s$-club.
Therefore, consider the vertex pair $\{a^v_i,a^v_j\}$ for some~$v\in K$. 
Observe that~$P\coloneqq (a^v_i,a^v_{i+1},\ldots, a^v_{i+p})$ for~$i+p=j$ is a path of length~$p$ from~$a^v_i$ to~$a^v_j$ and that~$Q\coloneqq (a^v_i,a^v_{i-1},\ldots, a^v_{i-q})$ for~$i-q=j$ is a path of length~$q$ from~$a^v_i$ to~$a^v_j$.
Clearly,~$p+q=x+1$.
Hence,~$\min(p,q)\le (x+1)/2\le 3\ell^*(s-1)+\lfloor\ell/2\rfloor$.
Without loss of generality, assume that the minimum is achieved by path~$P$ and assume that~$p=\alpha \cdot(3\ell^*)+\beta$ for some~$\alpha\in[s-1]$ and some~$\beta<3\ell^*$.
Recall that by Part~\ref{item-edge-variant-2} of Construction~\ref{const-ell-triangle-s-club-w-hardness},~$a^v_{i'}a^v_{j'}\in E(G')$ if and only if~$j'=i'+z$ for some~$z\in[-3\ell^*,3\ell^*]\setminus\{0\}$.
Hence, $$(a^v_i,a^v_{i+1\cdot(3\ell^*)}, \ldots, a^v_{i+\alpha\cdot(3\ell^*)}, a^v_{i+\alpha\cdot(3\ell^*)+\beta})$$ is a path of length at most~$(s-1)+1=s$ from~$a^v_i$ to~$a^v_j$.

These arguments also apply symmetrically to the vertex pairs~$\{b^v_i,b^v_j\}$ and $\{a^v_i,b^v_j\}$ for each~$v\in K$.
Furthermore, if~$\ell$ is odd, observe that the above argumentation can also be used to show that the vertex pairs~$\{c^v_i,a^v_j\}$,~$\{c^v_i,b^v_j\}$, and~$\{c^v_i,c^v_j\}$ have distance at most~$s$ to each other.

Second, we show that~$a^v_i$ has distance at most~$s$ to~$b^u_j$.
Note that by Part~\ref{item-edge-variant-7} of Construction~\ref{const-ell-triangle-s-club-w-hardness},~$a^v_i$ has neighbors~$b^u_i, \ldots , b^u_{i+\lfloor \ell/2\rfloor}$ since~$uv\in E(G)$.
In the following, we assume that~$j\ne i+z$ for all~$z\in[0,\lfloor \ell/2\rfloor]$.
Consider the paths~$P\coloneqq (a^v_i,b^u_{i+\lfloor \ell/2\rfloor},b^u_{i+\lfloor \ell/2\rfloor+1},\ldots, ,b^u_{i+\lfloor \ell/2\rfloor+p})$ for~$i+\lfloor \ell/2\rfloor+p= j$ of length~$p+1$ and~$Q\coloneqq (a^v_i,b^u_i,b^u_{i-1}, \ldots, b^u_{i-q})$ for~$i-q= j$ of length~$q+1$.
Observe that~$(p+1)+(q+1)=(x+3)-\lfloor \ell/2\rfloor$.
Thus,~$p+q=(x+1)-\lfloor \ell/2\rfloor=6\ell^*(s-1)+1$.
Since~$p$ and~$q$ are integers we have~$\min(p,q)\le ((x+1)-\lfloor \ell/2\rfloor)/2=3\ell^*(s-1)$.

Without loss of generality assume that the minimum is achieved by path~$P$ and assume that~$p=\alpha \cdot(3\ell^*)+\beta$ for some~$\alpha\in[s-2]$ and some~$\beta\le 3\ell^*$.
Recall that by Part~\ref{item-edge-variant-2} of Construction~\ref{const-ell-triangle-s-club-w-hardness}, we have~$b^u_{i'}b^u_{j'}\in E(G')$ if and only if~$j'=i'+z$ for some~$z\in[-3\ell^*,3\ell^*]\setminus\{0\}$.
Now, observe that $$(a^v_i,b^u_{i+\lfloor\ell/2\rfloor},b^u_{i+\lfloor\ell/2\rfloor+1\cdot(3\ell^*)}, \ldots, b^u_{i+\lfloor\ell/2\rfloor+\alpha\cdot(3\ell^*)}, b^u_{i+\lfloor\ell/2\rfloor+\alpha\cdot(3\ell^*)+\beta})$$ is a path of length at most~$1+(s-2)+1=s$ from~$a^v_i$ to~$b^u_j$.

Furthermore, if~$\ell$ is odd, observe that the above argumentation can also be used to show that the vertex pairs~$\{c^v_i,b^u_j\}$ have distance at most~$s$ to each other by replacing~$a^v_i$ with~$c^v_i$ in the paths~$P$ and~$Q$.

The fact that vertices~$a^v_i$ and~$a^u_j$, and~$b^v_i$ and~$b^u_j$, respectively, have distance at most~$s$ to each other can be proven similar as we showed that~$a^v_i$ and~$b^u_j$ have distance at most~$s$ by observing that~$a^v_i$ has distance~$2$ to each vertex~$a^u_{i+z}$ with~$z\in [-3\ell^*,\lfloor\ell/2\rfloor+3\ell^*]$ since~$a^v_i$ has neighbors~~$b^u_i,\ldots, b^u_{i+\lfloor\ell/2\rfloor}$ and since by Part~\ref{item-edge-variant-2} of Construction~\ref{const-ell-triangle-s-club-w-hardness}, we have~$b^u_{i'}a^u_{j'}\in E(G')$ if and only if~$j'=i'+z'$ for some~$z'\in[-3\ell^*,3\ell^*]$.

Furthermore, if~$\ell$ is odd, observe that the above argumentation can also be used to show that the vertex pairs~$\{c^v_i,a^u_j\}$ and~$\{c^v_i,c^u_j\}$ have distance at most~$s$ to each other by replacing~$a^v_i$ with~$c^v_i$ and replacing~$a^u_j$ with~$c^u_j$, respectively, in the paths~$P$ and~$Q$.

Hence,~$\widetilde{G}$ is indeed an~$s$-club.

\textbf{Next, we show that each edge in~$E(\widetilde{G})$ is contained in at least~$\ell$ triangles which are contained in~$\widetilde{G}$.}

Consider the edge~$a^v_ia^v_{i+j}$ for some~$j\in[-3\ell^*,3\ell^*]$.
Without loss of generality, assume that~$j>0$.
By Part~\ref{item-edge-variant-2} Construction~\ref{const-ell-triangle-s-club-w-hardness}, both vertices are adjacent to each vertex~$a^v_{i+i'}$ with~$i'\in[3\ell^*]\setminus\{j\}$. 
Hence, both vertices are in at least~$3\ell^*-1\ge\ell$ triangles.

Furthermore, the statement can be shown analogously for the edges~$b^v_ib^v_{i+j}$ (and $c^v_ic^v_{i+j}$ if~$\ell$ is odd) for some~$j\in[-3\ell^*,3\ell^*]$.

Also, the statement can be shown analogously for the edge~$a^v_ib^v_{i+j}$ for some $j\in[-3\ell^*,3\ell^*]$.
If~$j\ge 0$, then~$a^v_{i+z}$ for each~$z\in[3\ell^*]$ is a common neighbor of both vertices and thus the edge~$a^v_ib^v_{i+j}$ is contained in at least~$\ell$ triangles. 
The case~$j<0$ can be shown analogously.

For odd~$\ell$, the statement can be shown analogously for the edges $a^v_ic^v_{i+j}$, $b^v_ic^v_{i+j}$, and~$c^v_ic^v_{i+j}$ for some~$j\in[-3\ell^*,3\ell^*]$.

The fact that the edges~$a^v_ib^u_{i+j}$ and~$c^v_ib^u_{i+j}$ are contained in exactly~$\ell$ triangles follows from the proof of Lemma~\ref{lem-ell-triangle-s-club-av-also-bu}.

We conclude that each edge in~$E(\widetilde{G})$ is contained in at least~$\ell$ triangles in~$\widetilde{G}$.
Thus,~$S$ is indeed an edge-$\ell$-triangle~$s$-club of size~$k'$.

Conversely, let~$S$ be an edge-$\ell$-triangle~$s$-club of size at least~$k'$ in~$G'$.
More precisely, let~$\widetilde{G}$ be a maximal spanning subgraph of~$G[S]$ which has diameter at most~$s$ and such that each edge in~$E(\widetilde{G})$ is contained in at least~$\ell$ triangles is~$\widetilde{G}$.
We show that~$G$ contains a clique of size at least~$k$.

First, we show that for each vertex~$x\in A_v\cup B_v\cup C_v$ there exists a vertex~$y\in A_u\cup B_u\cup C_u$ such that~$\dist(x,y)\ge s+1$ if~$uv\notin E(G)$.
For this, recall that by construction each two vertices with sub-indices~$i'$ and~$j'$ are not adjacent if their difference (modulo~$x$) is larger than~$3\ell^*$.

\begin{claim}
\label{claim-ell-triangle-high-distance-uv-notin-eg'}
In~$G'$ we have~$\dist(x_i,y_{j})\ge s+1$ for each~$i\in[0,x]$,~$j\coloneqq i+\lfloor\ell/2\rfloor+3\ell^*(s-1)$,~$x_i\in\{a^v_i,b^v_i,c^v_i\}$, and~$y_j\in\{a^u_j,b^u_j,c^u_j\}$ if~$uv\notin E(G)$.
\end{claim}

\begin{claimproof}
There are two possible paths from~$x_i$ to~$y_{i+\lfloor\ell/2\rfloor+3\ell^*(s-1)}$ with respect to the indices.
First, there is a subsequence of the indices which is increasing ($i,i+1,\ldots,i+\lfloor\ell/2\rfloor+3\ell^*(s-1)$).
This path has length~$\lfloor\ell/2\rfloor+3\ell^*(s-1)$.
Second, there is a subsequence of the indices which is decreasing ($i,i-1,\ldots,i-i'$, where~$-i'=i+\lfloor\ell/2\rfloor+3\ell^*(s-1)$).
This path has length~$3\ell^*(s-1)+1$.

Hence, each path from~$x_i$ to~$y_{i+\lfloor\ell/2\rfloor+3\ell^*(s-1)}$ has to overcome at least~$3\ell^*(s-1)+1$ indices.
Observe that whenever an edge between~$A_p$ (or~$C_p$) and~$B_q$ for~$p,q\in V(G)$ with~$pq\in E(G)$ is traversed, by construction the index can increase/decrease by at most~$\lfloor\ell/2\rfloor$.
Now, we use the fact that~$uv\notin E(G)$:
There are no edges between the vertex gadgets~$T^v$ and~$T^u$.
Thus, at least two times such a traversal of at most~$\lfloor\ell/2\rfloor$ indices has to be done. 
Hence, the index~$i$ can increase or decrease by at most~$2\cdot\lfloor\ell/2\rfloor+3\ell^*(s-2)<3\ell^*(s-1)+1$ if at least~$2$ edge traversals between different vertex gadgets are necessary.
Thus, both vertices have distance at least~$s+1$.
\end{claimproof}

The following statement directly follows from Claim~\ref{claim-ell-triangle-high-distance-uv-notin-eg'} and Lemma~\ref{lem-ell-triangle-s-club-av-also-bu}.

\begin{claim}
\label{claim-ell-triangle-high-distance-uv-notin-eg'-consequenz}
If~$A_v\subseteq S$ (or~$C_v\subseteq S$ if~$\ell$ is odd) and~$B_u\subseteq S$  then~$uv\in E(G)$.
\end{claim}

We now use Claims~\ref{claim-ell-triangle-high-distance-uv-notin-eg'} and~\ref{claim-ell-triangle-high-distance-uv-notin-eg'-consequenz} to show that~$G$ contains a clique of size at least~$k$.
We distinguish the cases whether~$S$ contains only parts of one of the gadgets~$A_v$,~$B_v$, or~$C_v$ or whether~$S$ contains all vertices of the gadgets~$A_v$ (or~$B_v$, or~$C_v$) completely.

First, assume that for some vertex~$v\in V(G)$ we have~$A_v\cap S\ne\emptyset$ and~$A_v\not\subseteq S$.
In the following, we show that~$S$ only contains vertices of gadget~$T^v$ and from gadgets~$T^u$ such that~$uv\in E(G)$.
Since~$A_v\not\subseteq S$, we conclude that in~$\widetilde{G}$ we have~$N_{\widetilde{G}}(A_v\cap S)\subseteq (B_v\cup C_v)$: 
Otherwise, vertex~$a_v$ has a neighbor~$b_u\in B_u$ and by Lemma~\ref{lem-ell-triangle-s-club-av-also-bu} we would obtain~$A_v\subseteq S$, a contradiction to the assumption~$A_v\not\subseteq S$.
If~$B_v\not\subseteq S$, then by Lemma~\ref{lem-ell-triangle-s-club-av-also-bu}  no vertex in~$B_v$ can have a neighbor~$a^w_i$ or~$c^w_i$ for some~$w\ne v$.
Hence,~$S\cap T^v$ would be a connected component of size at most~$3(x+1)$, a contradiction to the size of~$S$ since~$k\ge 3$.
Thus, we may assume that~$B_v\subseteq S$.
Observe that if~$a^w_i\in S$ or~$c^w_i\in S$ for some~$w\in V(G)$ such that~$vw\in E(\widetilde{G})$, that is, also~$vw\in E(G')$, then we have~$A_w\subseteq S$ and~$C_w\subseteq S$ by Lemma~\ref{lem-ell-triangle-s-club-av-also-bu} since each vertex~$a^w_i$ and~$c^w_i$ has a neighbor in~$B_v$.
Let~$W\coloneqq \{w_1,\ldots, w_t\}$ denote the set of vertices~$w_j$  such that~$vw_j\in E(G)$  and~$A_{w_j},C_{w_j}\subseteq S$.
If~$w_xw_y\notin E(G)$ for some~$x,y\in[t]$ with~$x\ne y$, then~$a^{w_x}_0$ and~$a^{w_y}_{\lfloor\ell/2\rfloor+3\ell^*(s-1)}$ have distance at least~$s+1$ by Claim~\ref{claim-ell-triangle-high-distance-uv-notin-eg'}.
Thus~$w_xw_y\in E(G)$ for each~$x,y\in[t]$ with~$x\ne y$.

Assume towards a contradiction that~$a^p_i\in S$  or~$c^p_i\in S$ for some~$p\in V(G)\setminus w$ with~$p\ne v$.
Note that~$pv\notin E(G)$ since otherwise~$p\in W$ by the definition of~$W$.
Observe that since~$B_v\subseteq S$ we also have~$b^v_{i+\lfloor\ell/2\rfloor+3\ell^*(s-1)}\in S$.
But since~$pv\notin E(G)$ we obtain from Claim~\ref{claim-ell-triangle-high-distance-uv-notin-eg'} that~$\dist(z_i,b^v_{i+\lfloor\ell/2\rfloor+3\ell^*(s-1)})\ge s+1$ for~$z_i=a_i$ or~$z_i=c_i$, a contradiction.
We conclude that~$S$ does not contain any vertex~$a^p_i$ or~$c^p_i$ with~$p\ne v$ or~$p\ne w_j$ for~$j\in[t]$.

Next, assume towards a contradiction that~$b^p_i\in S$ for some~$p\in V(G)$ with $p\ne v$ and~$p\notin W$.
If~$pv\notin E(G)$, then~$b^p_i$ and~$b^v_{i+\lfloor\ell/2\rfloor+3\ell^*(s-1)}$ have distance at least~$s+1$ again by Claim~\ref{claim-ell-triangle-high-distance-uv-notin-eg'}.
Thus, we can assume that~$pv\in E(G)$.
Recall that~$b^v_{i+\lfloor\ell/2\rfloor+3\ell^*(s-1)}\in S$.
As defined by Claim~\ref{claim-ell-triangle-high-distance-uv-notin-eg'}, each shortest path from~$b^p_i$ to~$b^v_{i+\lfloor\ell/2\rfloor+3\ell^*(s-1)}$ can swap at most once between different vertex gadgets. 
 In this case, there is exactly one swap from~$T^p$ to~$T^v$.
From the above we know that~$(A_p\cup C_p)\cap S=\emptyset$.
Thus, each shortest path from~$b^p_i$ to~$b^v_{i+\lfloor\ell/2\rfloor+3\ell^*(s-1)}$ uses at least one vertex  in~$A_v\cup C_v$.
Since at least one edge with endpoint~$A_v\cup C_v$ is contained in~$S$, we conclude from  Lemma~\ref{lem-ell-triangle-s-club-av-also-bu}  that~$A_v\cup C_v\subseteq S$, a contradiction to the assumption~$A_v\not\subseteq S$.

 Hence, there is no vertex~$p\ne v$ and~$p\notin W$ such that~$T^p\cap S\ne\emptyset$.
 In other words,~$S$ contains only vertices from the gadget~$T^v$ and from gadgets~$T^u$ with~$vu\in E(G)$.
 Thus,~$S\subseteq T^v\cup\bigcup_{j=1}^t T^{w_j}$.
By definition of~$k'$, we have~$t\ge k-1$ and we conclude that~$G$ contains a clique of size at least~$k$.
The case that we have~$B_v\cap S\ne\emptyset$ and~$B_v\not\subseteq S$ or~$C_v\cap S\ne\emptyset$ and~$C_v\not\subseteq S$ for some vertex~$v\in V(G)$ can be handled similarly.

Second, consider the case that for each set~$A_v$ with~$A_v\cap S\ne\emptyset$ we have~$A_v\subseteq S$, that for each set~$B_v$ with~$B_v\cap S\ne\emptyset$ we have~$B_v\subseteq S$ , and that for each set~$C_v$ with~$C_v\cap S\ne\emptyset$ we have~$C_v\subseteq S$.
Let~$W_A\coloneqq \{w_A^j\mid A_{w_j}\subseteq S\}$,~$W_B\coloneqq \{w_B^j\mid B_{w_j}\subseteq S\}$ ,and~$W_C\coloneqq \{w_C^j\mid C_{w_j}\subseteq S\}$.
If~$W_A=\emptyset$ or~$W_B=\emptyset$  or~$W_C=\emptyset$then each connected component in~$G'[S]$ has size at most~$2(x+1)<k'$.
Thus, we may assume that~$W_A\ne\emptyset$,~$W_B\ne\emptyset$ ,and~$W_C\ne\emptyset$. 
By Claim~\ref{claim-ell-triangle-high-distance-uv-notin-eg'-consequenz}, we have~$w^i_Aw^j_B\in E(G)$ for each~$w^i_A\in W_A$ and~$w^j_B\in W_B$ and also~$w^i_Cw^j_B\in E(G)$ for each~$w^i_C\in W_C$ and~$w^j_B\in W_B$.
Furthermore, by Claim~\ref{claim-ell-triangle-high-distance-uv-notin-eg'}, we have~$w^j_Bw^{j'}_B\in E(G)$ for~$w^j_B,w^{j'}_B\in W_B$,~$w^i_Aw^{i'}_A\in E(G)$ for~$w^i_A,w^{i'}_A\in W_A$,~$w^i_Cw^{i'}_C\in E(G)$ for~$w^i_C,w^{i'}_C\in W_C$, and also ,~$w^i_Aw^{i'}_C\in E(G)$ for~$w^i_A\in W_A$ and~$w^{i'}_C\in W_C$.
Hence, we obtain that~$\min(|W_A|,|W_B|,|W_C|)\ge k$ and thus~$G$ contains a clique of size~$k$.
\qed
\end{proof}

\section{Seeded~\texorpdfstring{$s$}{s}-Club}

In this section we study the parameterized complexity of \textsc{Seeded~$s$-Club} with respect to the standard parameter solution size~$k$.
Recall that is this problem we aim to find an $s$-club containing a given seed of vertices.
Here, we assume that~$|W|<k$ since otherwise the problem can be solved in polynomial time.

\subsection{Tractable Cases} 

 For clique seeds, we provide the following kernel.
 Note that here we present a kernel and not only a Turing kernel.
 
\begin{theorem}
\label{thm-w-seeded-tractable-cases}
\textsc{Seeded~$s$-Club} admits a kernel with $\Oh(k^{2|W|+1})$ vertices if~$W$ is a clique. 
\end{theorem}

Note that this kernel has polynomial size if~$W$ has constant size.
In the following, assume that~$G[W]$ is a clique. 
 To prove the kernel, we first remove all vertices with distance at least~$s+1$ to any vertex in~$W$. 
Second, we show that if the remaining graph, that is,~$N_s[W]$, is sufficiently large, then~$(G,W,k)$ is a trivial yes-instance.

\begin{rrule}
\label{rr-w-seesed-s-club-remove-vertices-large-distance}
Let~$(G,W,k)$ be an instance of \textsc{Seeded~$s$-Club}. 
If~$G$ contains a vertex~$u$ such that~$\dist(u,w)\ge s+1$ for some~$w\in W$, then remove~$u$.
\end{rrule}

Clearly, Reduction Rule~\ref{rr-w-seesed-s-club-remove-vertices-large-distance} is correct and can be applied in polynomial time.
Next, we show that if the remaining graph is sufficiently large then~$(G,W,k)$ is a yes-instance of \textsc{Seeded~$s$-Club}.

\begin{lemma}
\label{lem-w-seeded-bound-of-ti-small}
An instance~$(G,W,k)$ of \textsc{Seeded~$s$-Club} with ~$|N_{s-1}[W]|\ge k^2$ is a yes-instance.
\end{lemma}

To prove Lemma~\ref{lem-w-seeded-bound-of-ti-small} for~$s\ge 3$ we need the following technical lemmas.

\begin{lemma}
\label{lem-seeded-s-club-clique-sge3-small}
An instance~$(G,W,k)$ of \textsc{Seeded~$s$-Club} with~$s\ge 3$ is a yes-instance if~$|N_{\lfloor (s+1)/2\rfloor-1}[W]|\ge k$.
\end{lemma}
\begin{proof}
By definition~$W\subseteq N_{\lfloor (s+1)/2\rfloor-1}[W]$ and~$|N_{\lfloor (s+1)/2\rfloor-1}[W]|\ge k$.
Thus, it remains to show that~$N_{\lfloor (s+1)/2\rfloor-1}[W]$ is an~$s$-club.
Therefore, consider a pair of vertices~$u, v\in N_{\lfloor (s+1)/2\rfloor-1}[W]$.
Observe that by definition~$\dist(u,W)\le \lfloor (s+1)/2\rfloor-1$ and~$\dist(v,W)\le \lfloor (s+1)/2\rfloor-1$.
Since~$W$ is a clique, we have~$\dist(u,v)\le (\lfloor (s+1)/2\rfloor-1)+1+(\lfloor (s+1)/2\rfloor-1)\le s$.
Hence, the lemma follows.\qed
\end{proof}

Note that the assumption~$s\ge 3$ in Lemma~\ref{lem-seeded-s-club-clique-sge3-small} is necessary to guarantee that~$\lfloor (s+1)/2\rfloor-1\ge 1$.
Next, we show that if a vertex in~$N_{\lfloor (s+1)/2\rfloor-1}(W)$ has many vertices close to it, then~$(G,W,k)$ is a yes-instance.

\begin{lemma}
\label{lem-seeded-s-club-clique-sge3-big}
An instance~$(G,W,k)$ of \textsc{Seeded~$s$-Club} in which~$s\ge 3$ and $|N_{\lfloor s/2\rfloor}(v)|\ge k$ for some vertex~$v\in N_{\lfloor (s+1)/2\rfloor-1}(W)$ is a yes-instance.
\end{lemma}
\begin{proof}
Let~$v$ be a vertex as specified in the lemma.
By definition of~$v$, there exists a path~$P\coloneqq (q_0,q_1,\ldots, q_{\lfloor (s+1)/2\rfloor-1})$ of length~$\lfloor (s+1)/2\rfloor-1$ in~$G$ such that~$q_{\lfloor (s+1)/2\rfloor-1}=v$,~$w_0\in W$, and~$q_i\in N_i(q_0)$.
We show that~$S\coloneqq N_{\lfloor s/2\rfloor}(v)\cup W\cup P$ is an~$s$-club of size~$k$ containing~$W$.
Clearly,~$W\subseteq S$ and~$|S|\ge k$.
Thus, it remains to show that~$S$ is an~$s$-club.

Consider a vertex~$w\in W$. 
Vertex~$w$ has distance at most~$i+1$ to vertex~$w_i$.
In particular,~$\dist(w,v)\le \lfloor (s+1)/2\rfloor$.
Since each vertex~$u\in N_{\lfloor s/2\rfloor}[v]$ has distance at most~$\lfloor s/2\rfloor$ to~$v$ we obtain that~$\dist(w,u)\le \lfloor (s+1)/2\rfloor+\lfloor s/2\rfloor=s$.
By similar arguments we can also show that vertex~$w_i$ for~$i\in [\lfloor (s+1)/2\rfloor-1]$ has distance at most~$s$ to each vertex in~$S$.

Finally, consider two vertices~$x,y\in N_{\lfloor s/2\rfloor}[v]$.
Note that~$\dist(x,v)\le \lfloor s/2\rfloor$ and also~$\dist(y,v)\le \lfloor s/2\rfloor$ and thus~$\dist(x,y)\le s$.

Thus,~$S$ is indeed an~$s$-club.\qed
\end{proof}

With those two lemmas we are now able to prove Lemma~\ref{lem-w-seeded-bound-of-ti-small}.

\begin{proof}[Proof of Lemma~\ref{lem-w-seeded-bound-of-ti-small}]
First, we consider the case~$s=2$.
It is sufficient to show that~$(G,W,k)$ is a yes-instance if~$|N[w]|\ge k$ for some~$w\in W$ since~$|W|\le k$ and since~$N[W]\ge k^2$ by our assumption.
Since all vertices in~$N(w)$ have the common neighbor~$w$, we conclude that~$N[w]$ is a~$2$-club.
Also, since~$W$ is a clique, we have~$W\subseteq N[w]$.
The size bound of~$N_{s-1}[W]$ follows from~$|N[w]|\ge k$.
Thus,~$(G,W,k)$ is a yes-instance.

Second, we consider the case~$s\ge 3$.
Observe that $$N_{s-1}[W]= N_{\lfloor (s+1)/2\rfloor-1}[W]\cup\bigcup_{v\in N_{\lfloor (s+1)/2\rfloor-1}[W]}N_{\lfloor s/2\rfloor}[v].$$

By Lemma~\ref{lem-seeded-s-club-clique-sge3-small},~$(G,W,k)$ is a yes-instance if~$|N_{\lfloor (s+1)/2\rfloor-1}[W]|\ge k$ and, by Lemma~\ref{lem-seeded-s-club-clique-sge3-big}, $(G,W,k)$ is a yes-instance if~$|N_{\lfloor s/2\rfloor}(v)|\ge k$ for some~$v\in N_{\lfloor (s+1)/2\rfloor-1}(W)$.
Thus, by the above equality we conclude that~$(G,W,k)$ is a yes-instance.\qed
\end{proof}

Finally, we bound the size of~$N_s(W)$.
There we assume that~$|N_{s-1}[W]|<k^2$ by Lemma~\ref{lem-w-seeded-bound-of-ti-small} and that Reduction Rule~\ref{rr-w-seesed-s-club-remove-vertices-large-distance} is applied.

\begin{lemma}
\label{lem-w-seeded-bound-of-ti-s}
An instance~$(G,W,k)$ of \textsc{Seeded~$s$-Club} with~$|N_s(W)|\ge k^{2|W|+1}$ which is reduced with respect to Reduction Rule~\ref{rr-w-seesed-s-club-remove-vertices-large-distance}  is a yes-instance.
\end{lemma}
\begin{proof}
Since Reduction Rule~\ref{rr-w-seesed-s-club-remove-vertices-large-distance} has been  applied exhaustively, each vertex~$p\in N_s(W)$ has distance exactly~$s$ to each vertex in~$W$.
In other words, for each vertex~$w_\ell\in W$ there exists a vertex~$u^{\ell}_{s-1}\in N_{s-1}(w_\ell)$ such that~$pu^\ell_{s-1}\in E(G)$.
Note that~$N_{s-1}(w_\ell)\subseteq N_{s-1}[W]$.
Moreover, by Lemma~\ref{lem-w-seeded-bound-of-ti-small} we may assume that $|N_{s-1}[W]|<k^2$.
In particular:~$|N_{s-1}(W)|<k^2$.
Since~$|N_s(W)|\ge k^{2|W|+1}$, by the pigeonhole principle there exists a set~$\{u^1_{s-1},u^2_{s-1},\ldots, u^{|W|}_{s-1}\}$ with $u^\ell_{s-1}\in N_{s-1}(w_\ell)$ for~$\ell\in[|W|]$ such that the set $P\coloneqq N_s(W)\cap\bigcap_{\ell\in[|W|]}N(u^\ell_{s-1})$ has size at least~$k$. 
The size bound follows from the observation that each~$N_{s-1}(w_\ell)$ has size at most~$k^2$ and we have exactly~$|W|$ many of these sets.
By the definition of vertex~$u^\ell_{s-1}$, there exists for each~$i\in[s-2]$ a vertex~$u_i^\ell\in N_i(w_\ell)$ such that~$w_\ell, u_1^\ell, \ldots , u_{s-1}^\ell$ is a path of length~$s-1$ in~$G$.
We define the set~$U\coloneqq \{u_i^\ell\mid\ell\in[|W|], i\in[s-1]\}$.
Next, we show that~$Z\coloneqq P\cup W\cup U$ induces an~$s$-club.

First, observe that all vertices in~$P$ have distance at most~$2$ to each other since they have the common neighbor~$u^1_{s-1}$.
Second, note that the vertices~$w_\ell$, $u^\ell_1,\ldots, u^\ell_{s-1}$,~$p,u^j_{s-1},\ldots, u^j_1,w_j$ form a cycle with~$2s+1$ vertices, for each~$p\in P$ and each two indices~$j,\ell\in[|W|]$.
Each vertex in this cycle has distance at most~$s$ to each other vertex in that cycle.
Hence,~$Z$ is indeed an~$s$-club.\qed
\end{proof}

Recall that Lemma~\ref{lem-w-seeded-bound-of-ti-small} showed that the number of vertices with distance \emph{at most}~$s-1$ to~$W$ is bounded by~$k^2$.
Together with Lemma~\ref{lem-w-seeded-bound-of-ti-s} now Theorem~\ref{thm-w-seeded-tractable-cases} is proven.

\subsection{Intractable Cases}
Now, we show hardness for some of the remaining cases.

\begin{theorem}
Let~$H$ be a fixed graph. \textsc{Seeded~$s$-Club} is W[1]-hard parameterized by~$k$ even if~$G[W]$ is isomorphic to~$H$, when
\begin{itemize}
\item $s=2$ and~$H$ contains at least two non-adjacent vertices, or if
\item $s\ge 3$ and~$H$ contains at least two connected components.
\end{itemize}
\end{theorem}

\paragraph{Hardness for $s=2$.} First, we prove hardness for~$s=2$ when~$H$ contains at least one non-edge.
For an illustration of Construction~\ref{const-w-seeded-s-club-s-equals-two} we refer to Figure~\ref{fig-examples-seeded-club}.

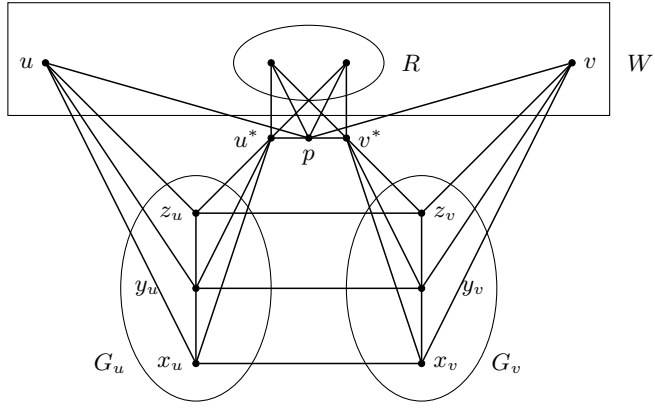
\begin{figure}[t]
\centering
\begin{tikzpicture}
\node[label=left:{$x_u$}](xu) at (0, 0) [shape = circle, draw, fill=black, scale=0.07ex]{};
\node[label=left:{$y_u$}](yu2) at (-0.25, 1) {};
\node(yu) at (0, 1) [shape = circle, draw, fill=black, scale=0.07ex]{};
\node[label=left:{$z_u$}](zu) at (0, 2) [shape = circle, draw, fill=black, scale=0.07ex]{};
\path [-,line width=0.2mm](xu) edge (yu);
\path [-,line width=0.2mm](zu) edge (yu);
\draw (0,1) ellipse (1cm and 1.5cm);
\node[label=left:{$G_u$}](gg) at (-0.7, 0) {};

\node[label=right:{$x_v$}](xv) at (3, 0) [shape = circle, draw, fill=black, scale=0.07ex]{};
\node[label=right:{$y_v$}](yv2) at (3.3, 1) {};
\node(yv) at (3, 1) [shape = circle, draw, fill=black, scale=0.07ex]{};
\node[label=right:{$z_v$}](zv) at (3, 2) [shape = circle, draw, fill=black, scale=0.07ex]{};
\path [-,line width=0.2mm](xv) edge (yv);
\path [-,line width=0.2mm](zv) edge (yv);
\draw (3,1) ellipse (1cm and 1.5cm);
\node[label=right:{$G_v$}](gg) at (3.7, 0) {};

\path [-,line width=0.2mm](xu) edge (xv);
\path [-,line width=0.2mm](yu) edge (yv);
\path [-,line width=0.2mm](zu) edge (zv);

\node[label=below:{$p$}](p) at (1.5, 3) [shape = circle, draw, fill=black, scale=0.07ex]{};
\node[label=left:{$u^*$}](us) at (1, 3) [shape = circle, draw, fill=black, scale=0.07ex]{};
\node[label=right:{$v^*$}](vs) at (2, 3) [shape = circle, draw, fill=black, scale=0.07ex]{};
\path [-,line width=0.2mm](p) edge (us);
\path [-,line width=0.2mm](p) edge (vs);
\path [-,line width=0.2mm](xu) edge (us);
\path [-,line width=0.2mm](yu) edge (us);
\path [-,line width=0.2mm](zu) edge (us);
\path [-,line width=0.2mm](xv) edge (vs);
\path [-,line width=0.2mm](yv) edge (vs);
\path [-,line width=0.2mm](zv) edge (vs);

\node[label=left:{$u$}](u) at (-2, 4) [shape = circle, draw, fill=black, scale=0.07ex]{};
\node[label=right:{$v$}](v) at (5, 4) [shape = circle, draw, fill=black, scale=0.07ex]{};
\path [-,line width=0.2mm](xu) edge (u);
\path [-,line width=0.2mm](yu) edge (u);
\path [-,line width=0.2mm](zu) edge (u);
\path [-,line width=0.2mm](xv) edge (v);
\path [-,line width=0.2mm](yv) edge (v);
\path [-,line width=0.2mm](zv) edge (v);
\draw (1.5,4) ellipse (1cm and 0.5cm);
\node[label=right:{$R$}](R) at (2.5, 4) {};
\node(r1) at (1, 4) [shape = circle, draw, fill=black, scale=0.07ex]{};
\node(r2) at (2, 4) [shape = circle, draw, fill=black, scale=0.07ex]{};
\path [-,line width=0.2mm](p) edge (u);
\path [-,line width=0.2mm](p) edge (v);
\path [-,line width=0.2mm](p) edge (r1);
\path [-,line width=0.2mm](p) edge (r2);
\path [-,line width=0.2mm](us) edge (r1);
\path [-,line width=0.2mm](us) edge (r2);
\path [-,line width=0.2mm](vs) edge (r1);
\path [-,line width=0.2mm](vs) edge (r2);
\draw[draw=black] (-2.5,3.3) rectangle ++(8,1.5);
\node[label=right:{$W$}](W) at (5.5, 4) {};
\end{tikzpicture}
\caption{Illustration of Construction~\ref{const-w-seeded-s-club-s-equals-two}.}
\label{fig-examples-seeded-club}
\end{figure}

\begin{const}
\label{const-w-seeded-s-club-s-equals-two}
Let~$(G,k)$ be an instance of \textsc{Clique}. 
We construct an equivalent instance~$(G',k')$ of \textsc{Seeded~$s$-Club} as follows.
Initially, we add the set~$W$ to~$G'$, and add edges such that~$G'[W]$ is isomorphic to~$H$. 
Since~$H$ is not a clique, there exist two vertices~$u,v\in V(H)$ such that~$uv\notin E(H)$.
Let~$R\coloneqq W\setminus\{u,v\}$.
Next, we add two copies~$G_u$ and~$G_v$ of~$G$ to~$G'$, make~$u$ adjacent to each vertex in~$V(G_u)$, and make~$v$ adjacent to each vertex in~$V(G_v)$. 
We denote with~$x_u$, and~$x_v$ the copies of~$x$ in~$G_u$ and~$G_v$, respectively.
Next, we add the edge~$x_ux_v$ for each~$x\in V(G)$.
Furthermore, we add a new vertex~$p$ and make it adjacent to each vertex in~$W$.
Next, we add a new vertex~$u^*$ adjacent to~$p$, each vertex in~$V(G_u)$, and each vertex in~$R$.
Analogously, we add a new vertex~$v^*$ which is adjacent to~$p$, each vertex in~$V(G_v)$, and each vertex in~$R$. 
Finally, we set~$k'\coloneqq 2k+|W|+3$.
\end{const}

Now, we prove the correctness of Construction~\ref{const-w-seeded-s-club-s-equals-two}.

\begin{lemma}
\label{lemma-w-seeded-2-club-hardness}
For any graph~$H$ which is not a clique, \textsc{Seeded~$2$-Club} parameterized by~$k$ is W[1]-hard if~$G'[W]$ is isomorphic to~$H$.
\end{lemma}

\begin{proof}
We prove that~$G$ contains a clique of size~$k$ if and only if~$G'$ contains a~$2$-club~$S$ containing~$W$ of size~$k'=2k+|W|+3$.

Let~$K$ be a clique of size~$k$ in~$G$ and let~$K_u$ and~$K_v$ be the copies of~$K$ in~$G_u$ and~$G_v$.
We argue that~$S\coloneqq K_u\cup K_v\cup W\cup \{u^*,v^*,p\}$ is a~$2$-club of size at least~$k'$ containing~$W$. 
Clearly,~$|S|=k'$ and~$S$ contains~$W$.
Thus, it remains to show that~$G[S]$ is a~$2$-club.

First, we show that each vertex in~$R$ has distance at most~$2$ to each vertex in~$S$:
All vertices in~$R$ have the common neighbors~$p$,~$u^*$, and~$v^*$. 
Since~$u$ and~$v$ are neighbors of~$p$, each vertex in~$V(G_u)$ is a neighbor of~$u^*$, and since each vertex in~$V(G_v)$ is a neighbor of~$v^*$ we conclude that each vertex in~$R$ has distance at most~$2$ to any vertex in~$S$.

Second, we show that each vertex in~$K_u$ has distance at most~$2$ to each vertex in~$S\setminus R$.
Observe that~$\{u,u^*,x_v\}\cup K_u\subseteq N[x_u]$ for each vertex~$x_u\in K_u$.
Hence,~$x_u$ has distance at most~$2$ to~$p$ via~$u^*$,~$v$, and~$v^*$ via~$x_v$, each vertex, and each vertex in~$K_v$ via the corresponding vertex in~$K_u$.
By symmetric arguments the statement also holds for each vertex in~$K_v$.

Finally, each pair of vertices of~$\{p,u,u^*,v,v^*\}$ has distance at most~$2$ to each other since~$u,u^*,v,v^*\in N(p)$.

Thus,~$S$ is indeed a~$2$-club.

Conversely, suppose that~$G'$ contains an~$2$-club~$S$ of size at least~$2k+|W|+3$ which contains all vertices of~$W$.
Observe that we have~$N(x_u)\cap N(v)=\{x_v\}$ for each vertex~$x_u\in V(G_u)$, and symmetrically~$N(x_v)\cap N(u)=\{x_u\}$ for each vertex~$x_v\in V(G_v)$.
Hence,~$x_u\in S$ if and only if~$x_v\in S$.
Let~$K_u\coloneqq S\cap V(G_u)$.
By definition of~$k'$ we obtain that~$|K_u|\ge k$.
Assume towards a contradiction that~$K_u$ contains a pair of nonadjacent vertices~$x_u$ and~$y_u$.
By the argumentation above we obtain~$y_v\in S$.
Now, observe that~$N(x_u)=\{u,u^*,x_v\}\cup\bigcup\{z_u\mid xz\in E(G)\}$ and that~$N(y_v)=\{v,v^*,y_u\}\cup\bigcup\{z_v\mid yz\in E(G)\}$.
Since~$xy\notin E(G)$ we thus obtain~$N(x_u)\cap N(y_v)=\emptyset$, a contradiction.
Thus,~$G$ contains a clique of size~$k$.\qed
\end{proof}

\paragraph{Hardness for seeds with at least two connected components and~$s\ge 3$.}
Now, we show W[1]-hardness for the case~$s\ge 3$ when the seed contains at least two connected components.
Fix a graph~$H$ with at least two connected components.
We show W[1]-hardness for~$s\ge 3$ even if~$G[W]$ is isomorphic to~$H$.

\begin{const}
\label{const-w-seeded-s-club-s-atleast3-many-conn-compos}
Let~$(G,k)$ be an instance of \textsc{Clique}. 
We construct an equivalent instance~$(G',k')$ of \textsc{Seeded~$s$-Club} as follows.
Initially, we add the set~$W$ to~$G'$, and add edges such that~$G'[W]$ is isomorphic to~$H$. 
Let~$D_1$ be one connected component of~$G'[W]$.
By assumption,~$D_2\coloneqq W\setminus D_1$ is not empty.
Next, we add two copies~$G_1$ and~$G_2$ of~$G$ to~$G'$. Then, we add edges to~$G'$ such that each vertex in~$D_1$ is adjacent to each vertex in~$V(G_1)$ and such that each vertex in~$D_2$ is adjacent to each vertex in~$V(G_2)$. 
Furthermore, we add a path~$(p_1, \ldots , p_{s-1})$ consisting of exactly~$s-1$ new vertices to~$G'$, make~$p_1$ adjacent to each~$u\in D_1$, and make~$p_{s-1}$ adjacent to each~$v\in D_2$.
By~$P\coloneqq \{p_i\mid i\in[s-1]\}$ we denote the set of these newly added vertices.
Now, for each~$x\in V(G)$ we do the following.
Consider the copies~$x_1\in V(G_1)$ and~$x_2\in V(G_2)$ of vertex~$x\in V(G)$.
We add a path~$(x_1, q_1^x, \ldots , q_{s-2}^x, x_2)$ consisting of~$s-2$ new vertices to~$G'$.
By~$Q_x\coloneqq \{q^x_i\mid i\in[s-2]\}$ we denote the set of the new internal path vertices.
Finally, we set~$k'\coloneqq sk+|W|+s-1$.
\end{const}

Now, we prove the correctness of Construction~\ref{const-w-seeded-s-club-s-atleast3-many-conn-compos}.

\begin{lemma}
\label{lem-w-seeded-s-club-s-atleast3-many-conn-compos}
Let~$H$ be a fixed graph with at least two connected components.
\textsc{Seeded~$s$-Club} parameterized by~$k$ is W[1]-hard even if~$G[W]$ is isomorphic to~$H$.
\end{lemma}
\begin{proof}
We show that~$G$ contains a clique of size~$k$ if and only if~$G'$ contains a~$W$-seeded~$s$-club of size at least~$k'=sk+|W|+s-1$.

Let~$K$ be a clique of size~$k$ in~$G$. 
Furthermore, let~$K_1$ and~$K_2$ denote the copies of~$K$ in~$G_1$ and~$G_2$, respectively.
We show that~$S\coloneqq W\cup P\cup K_1\cup K_2\cup\bigcup_{x\in K}Q_x$ is a~$W$-seeded~$s$-club of size at least~$k'$.
Clearly,~$|S|=k'$ and~$S$ contains~$W$.
Thus, it remains to verify that~$G[S]$ is an~$s$-club.
Note that since each vertex in~$V(G_1)$ is adjacent to each vertex in~$D_1$, each two vertices in~$D_1$ have distance at most~$2$, and similarly each two vertices in~$V(G_1)$ have distance at most~$2$.
Analogously, we can show that each two vertices in~$D_2$ and each two vertices in~$V(G_2)$ have distance at most~$2$.
Furthermore, the vertices~$(x,p_1,p_2,\ldots,p_{s-1},y,u_2,q^u_{s-2},\ldots, q^u_1,u_1)$ for each vertex~$x\in D_1$, each vertex~$y\in D_2$, and each vertex~$u\in K$ form a~$C_{2s+1}$, a cycle with~$2s+1$ vertices.
Observe that also the vertices~$(u_1,q^u_1,\ldots,q^u_{s-2},u_2,v_2,q^v_{s-2},\ldots, q^v_1,v_1)$ form a~$C_{2s}$ for each two vertices~$u,v\in K$.
Since all remaining distances are covered by these two cycles, we conclude that~$S$ is indeed an~$s$-club.

Conversely, suppose that~$G'$ contains a~$W$-seeded~$s$-club~$S$ of size at least~$k'$.
Let~$Q'_v\coloneqq \{v_1,q^v_1,\ldots, q^v_{s-2},v_2\}$ for each~$v\in V(G)$.
We show that~$Q'_v\cap S\ne\emptyset$ if and only if~$Q'_v\subseteq S$.
Assume towards a contradiction, that~$Q'_v\cap S\ne\emptyset$ for some~$v\in V(G)$ such that~$Q'_v\not\subseteq S$.
If~$v_1\notin S$, and also~$v_2\notin S$, then no vertex in~$S\cap Q'_v$ is connected to any vertex in~$S\setminus Q'_v$.
Hence, we can assume without loss of generality that~$v_1\in S$.
Note that~$N(D_2)=V(G_2)\cup\{p_{s-1}\}$.
Furthermore, observe that~$\dist(v_1,p_{s-1})=s$, that~$\dist(v_1,q^v_{s-2})=s-2$, and that~$\dist(v_1,q^u_{s-2})\ge s-1$ for each~$u\in V(G)\setminus\{v\}$.
Thus, the unique path of length at most~$s$ from~$v_1$ to~$D_2$ contains all vertices in~$Q'_v$.
Hence,~$Q'_v\cap S\ne\emptyset$ if and only if~$Q'_v\subseteq S$.
By the definition of~$k'$ we may thus conclude that~$Q'_v\subseteq S$ for at least~$k$ vertices~$v\in V(G)$.

Now, assume towards a contradiction that~$Q'_u\subseteq S$ and~$Q'_v\subseteq S$ such that~$uv\notin E(G)$.
We consider the vertices~$v_1$ and~$u_2$.
Observe that by construction each path from~$v_1$ to~$u_2$ containing any vertex~$p_i$ has length at least~$s+1$.
Hence, each shortest path from~$v_1$ to~$u_2$ contains the vertex set of~$Q'_w$ for some~$w\in V(G)$.
Since the path induced by each~$Q'_w$ has length~$s-1$, we conclude that~$w=u$ or~$w=v$.
Assume without loss of generality that~$w=v$.
Hence, the~$(s-1)$th vertex on the path from~$v_1$ is~$v_2$.
Since~$uv\notin E(G)$ we have by construction that~$u_2v_2\notin E(G')$.
Hence,~$\dist(v_1,u_2)\ge s+1$, a contradiction to the fact that~$S$ is an~$s$-club.
Thus,~$\{v\mid Q'_v\subseteq S\}$ is a clique of size at least~$k$ in~$G$.\qed
\end{proof}

\section{Conclusion}

We provided a complexity dichotomy for \trisclub{} and \etrisclub{} for the standard parameter solution size~$k$ with respect to~$s$ and~$\ell$.
Furthermore, we also provided a complexity dichotomy for \textsc{Seeded~$2$-Club} for~$k$ in terms of the structure of~$G[W]$.
For \textsc{Seeded~$s$-Club} with~$s\ge 3$ we provided an FPT-algorithm with respect to~$k$ when~$G[W]$ is a clique and we showed W[1]-hardness for~$k$ when~$G[W]$ contains at least 2 connected components.
Hence, an immediate open question is the parameterized complexity of \textsc{Seeded~$s$-Club} for~$s\ge3$ when~$G[W]$ is connected but not a clique.
One aim should be to also provide a dichotomy for~$k$ for \textsc{Seeded~$s$-Club} with~$s\ge 3$ for all possible structures of~$G[W]$.
It is particularly interesting to study seeds of constant size since this seems to be the most interesting case for applications.

For future work, it seems interesting to study the complexity of the considered variants of \textsc{$s$-Club} with respect to further parameters, for example with respect to structural parameters of the input graph~$G$ such as the treewidth of~$G$. 
Additionally, the parameterized complexity of further robust variants of \textsc{$s$-Club} such as~\textsc{$t$-Hereditary $s$-Club}~\cite{PYB13,KNNP19} with respect to~$k$ remains open.
It is also interesting to study other problems for detecting communities with seed constraints.
One prominent example is \textsc{$s$-Plex}.
This problem is also NP-hard for~$W\ne\emptyset$, since an algorithm for the case when~$|W|=1$ can be used as a black box to solve the unseeded variants.
From a practical perspective, we plan to implement combinatorial algorithms for all three problem variants for the most important special case~$s=2$. 
Based on experience with previous implementations for \textsc{$2$-Club}~\cite{HKN15} and some of its robust variants~\cite{KNNP19} we are optimistic that these problems can be solved efficiently on sparse real-world instances.

\bibliographystyle{splncs04}

\end{document}